\documentclass[11pt, letterpaper]{article}
\usepackage{amssymb, amsmath, amsthm}

\usepackage[english]{babel} 
\usepackage[utf8]{inputenc} 
\usepackage[T1]{fontenc} 
\usepackage{hyperref}
\usepackage[capitalise]{cleveref}
\usepackage[margin=1in]{geometry}
\usepackage{enumerate}
\usepackage{todonotes}
\usepackage{enumitem}
\usepackage[ruled,linesnumbered, noline,boxed,commentsnumbered, noend]{algorithm2e}
\usepackage{tikz-cd}
\usepackage{float}
\usepackage{thmtools,thm-restate}

\Crefname{algocf}{Algorithm}{Algorithms}

\title{Optimal Decremental Connectivity in Non-Sparse Graphs}
\author{Anders Aamand\\MIT \\aamand@mit.edu \and Adam Karczmarz\\University of Warsaw\\a.karczmarz@mimuw.edu.pl \and Jakub Łącki\\Google Research\\
jlacki@google.com \and Nikos Parotsidis\\Google Research\\nikosp@google.com  \and Peter M.\ R.\ Rasmussen\\University of Copenhagen\\pmrr@di.ku.dk \and Mikkel Thorup\\University of Copenhagen\\mikkel2thorup@gmail.com}
\date{\today}

\newcommand{\abs}[1]{\left\vert #1\right\vert}

\newcommand{\E}[1]{\mathbb E\left[ #1 \right]}

\newcommand{\Var}[1]{\text{Var}\left[ #1 \right]}

\newcommand{\N}{\mathbb{N}}

\newcommand{\eps}{\varepsilon}

\newcommand{\funcSpace}{0.1cm}

\def\poly{\operatorname{poly}}
\def\polylog{\operatorname{polylog}}

\newcommand{\Ot}{\widetilde{O}}

\newtheorem{theorem}{Theorem}[section]

\newtheorem{lemma}[theorem]{Lemma}

\newtheorem{definition}[theorem]{Definition}

\begin{document}

\maketitle
\thispagestyle{empty}

\begin{abstract}

We present a dynamic algorithm for maintaining the connected and 2-edge-connected components in an undirected graph subject to edge deletions.
The algorithm is Monte-Carlo randomized and processes any sequence of edge deletions in $O(m + n \poly \log n)$ total time.  Interspersed with the deletions, it can answer queries to whether any two given vertices currently belong to the same (2-edge-)connected component in constant time.
Our result is based on a general Monte-Carlo randomized reduction from decremental $c$-edge-connectivity to a variant of fully-dynamic $c$-edge-connectivity on a sparse graph.

For non-sparse graphs with $\Omega(n \poly \log n)$ edges, our connectivity and $2$-edge-connectivity algorithms handle all deletions in optimal linear total time, using existing algorithms for the respective fully-dynamic problems.
This improves upon an $O(m \log (n^2 / m) + n \poly \log n)$-time algorithm of Thorup [J.Alg. 1999], which runs in linear time only for graphs with $\Omega(n^2)$ edges.

Our constant amortized cost for edge deletions in decremental connectivity in non-sparse graphs should be contrasted with an $\Omega(\log n/\log\log n)$ worst-case time lower bound in the decremental setting [Alstrup, Thore Husfeldt, FOCS'98] as well as  an $\Omega(\log n)$ amortized time lower-bound in the fully-dynamic setting [Patrascu and Demaine STOC'04].
\end{abstract}

\clearpage
\setcounter{page}{1}

\section{Introduction}
In this paper, we present Monte Carlo randomized decremental dynamic
algorithms for maintaining the connected and 2-edge-connected components in an
undirected graph subject to edge deletions.  Starting from a graph
with $n$ nodes and $m$ edges, the algorithm can process any sequence of edge
deletions in $O(m + n \polylog n)$ total time while answering queries 
whether a pair of vertices is currently in the same (2-edge-)connected component. Each query is answered in constant time.  The algorithm for decremental 2-edge-connectivity additionally reports all bridges as they appear.

For some concrete constant $\alpha \leq 7$,
our decremental algorithms thus use $O(m)$ total
time on the edge deletions from a graph
with $m\geq n\log^\alpha n$ edges, and we will
refer to such graphs as \emph{non-sparse}.
If we delete all the edges, we support edge deletions in constant amortized time. As we shall discuss in Section \ref{sec:connectivity}, it is not possible to obtain a constant worst-case time bound for individual edge deletions in this decremental setting, nor is it possible to obtain a
constant amortized time bound for edge updates in the fully-dynamic version of the connectivity problem.

Our algorithms are Monte Carlo randomized and answer all queries correctly with high probability\footnote{We define high probability as probability $1-O(n^{-\gamma})$ for any given $\gamma$.}. We note that since the correct answer to each query is uniquely determined from the input, the algorithms work against adaptive adversaries, that is, each deleted edge may depend on previous answers to queries\footnote{Indeed, the only way to strengthen the adversary is to reveal information about the internal choices of the algorithm, through responses to queries.}.

Furthermore, our algorithms offer a \emph{self-check} capability.
At the end, after all updates and queries have been processed online, each algorithm can deterministically check if it might have made a mistake.
If the self-check passes, it is \emph{guaranteed} that no incorrect answer was given. Otherwise, the algorithm \emph{may have} made a mistake.
However, as we show in the following, the self-check passes with high probability.
This feature
implies that we can obtain Las Vegas algorithms for certain \emph{non-dynamic} problems whose solutions employ decremental (2-edge-)connectivity algorithms as subroutines: we
simply repeat trying to solve the static problem from scratch, each time with new random bits, until the final self-check is passed. 
With high probability, we are done already
after the first round. A  nice concrete example is the algorithm of 
 Gabow, Kaplan, and Tarjan~\cite{GabowKT01} for the static problem of deciding if a
graph has a unique perfect matching. The algorithm
uses a decremental  2-edge-connectivity 
algorithm as a subroutine. With our
decremental 2-edge-connectivity algorithm,
repeating until the self-check is passed, 
we obtain a Las Vegas
algorithm for the unique perfect matching problem that is always correct, and which terminates in $O(m + n
\polylog n)$ time with high probability.

The tradition of looking for linear time algorithms for
non-sparse graphs goes back at least to Fibonacci heaps, which can be used for solving single source shortest paths in $O(m+n\log n)$ time~\cite{FredmanT87}. Our results show that another fundamental graph problem can be solved in linear time in the non-sparse case.

The previous best time bounds for the decremental connectivity and 2-edge-connectivity problems
were provided by Thorup \cite{thorup1999decremental}. His algorithms run in $O(m \log (n^2/ m) + n \polylog n)$ total time. This is $O(m)$ only for dense graphs with $\Omega(n^2)$ edges. It should be noted that \cite{thorup1999decremental} used Las Vegas randomization, that is,
correctness was guaranteed, but the running time bound only held with high probability. Our algorithms are Monte Carlo randomized, but offer the final self-check.
Another difference is that our new algorithms need only a polylogarithmic number of random bits, whereas the ones from \cite{thorup1999decremental} used $\Theta(m)$ random bits.

Both our algorithm and the previous one by Thorup are based on a
general reduction from decremental $c$-edge-connectivity to fully-dynamic $c$-edge-connectivity on a sparse graph
with $\tilde O(cn)$ updates. The reductions have a polylogarithmic
cost per node as well as a cost per original edge. Our contribution is to
reduce the edge cost from $O(\log (n^2/m))$ to the optimal $O(1)$. The general reduction and its consequences will be discussed in Section \ref{sec:reduction-result}.

We will now give a more detailed discussion
of our results in the context of related
work.

\subsection{Connectivity}\label{sec:connectivity}

 Dynamic connectivity is the most fundamental dynamic graph problem. The fully dynamic version has been extensively studied~\cite{abs-1910-08025, EppsteinGIN97, Frederickson85,
  HenzingerK99, holm2001poly, HuangHKP17, KapronKM13,
  Kejlberg-Rasmussen16, NanongkaiSW17,patrascu2004lower, PatrascuT11a,
  Thorup00, Wang15w, Wulff-Nilsen13a} from both the lower and upper
bound perspective, even though close to optimal amortized update
bounds have been known since the 90s~\cite{HenzingerK99, holm2001poly,
  Thorup00}.  Currently, the best known amortized update time bounds
are (expected) $O(\log{n}\cdot (\log{\log{n}})^2)$ using
randomization~\cite{HuangHKP17} and $O(\log^2{n}/\log\log{n})$
deterministically~\cite{Wulff-Nilsen13a}.  Both these results have
optimal (wrt. to the lower bounds) query time.

\paragraph{Connectivity Lower Bounds}
Our result implies that decremental connectivity is provably easier than fully-dynamic connectivity for a wide range of graph densities.
Specifically, let $t_u$ be the update time of a fully dynamic
connectivity algorithm and let $t_q$ be its query
time. P\v{a}tra\c{s}cu and Demaine \cite{patrascu2004lower} showed a
lower bound of $\Omega(\log n)$ on $\max(t_u, t_q)$ in the cell-probe
model.
P\v{a}tra\c{s}cu and
Thorup~\cite{PatrascuT11a} also showed that $t_u=o(\log{n})$ implies
$t_q=\Omega(n^{1-o(1)})$.  These lower bounds hold for all graph
densities and allow for both amortization and randomization. As a
result, no fully-dynamic connectivity algorithm can
answer connectivity queries in constant time and have an amortized update time of $o(\log n)$.

In sharp contrast, assuming that $m = \Omega(n \poly \log n)$ edges are deleted, our algorithm shows that one can solve decremental connectivity handling both queries and updates in constant amortized time.

We note that such a result is possible only because we allow for amortization, as any decremental connectivity algorithm with \emph{worst-case} update time $O(\polylog n)$ must have worst-case query time $\Omega\left(\frac{\log n}{\log   \log n}\right)$~\cite{Alstrup1998}. This lower bound holds even for trees supporting restricted connectivity queries of the form "are $u$ and $v$ connected?" for a fixed ``root'' $u$. This lower bound also holds for dense graphs,
as we can always add a large static clique to the problem.

An optimal incremental connectivity algorithm has been known for over 40 years.
Namely, to handle $m\geq n$ edge insertion and $q$ connectivity queries, one can use the union-find data structure~\cite{tarjan1975efficiency} with $n - 1$ unions and $2(m+q)$ finds. The total running time is $\Theta((m+q)\alpha((m+q),n))$, which is linear for all but very sparse
graphs (since $\alpha(\Omega(n \log n), n) = O(1)$).
It was later shown that this running time is optimal for incremental
connectivity~\cite{Fredman89}.

Similarly to the decremental case, one cannot hope to obtain an analogous result with a worst-case update time in the incremental setting: P\v{a}tra\c{s}cu and
Thorup~\cite{PatrascuT11a} showed that any incremental connectivity
data structure with $o\left(\frac{\log n}{\log \log n}\right)$
worst-case update time must have worst-case $\Omega(n^{1-o(1)})$ query
time in the cell-probe model.

\paragraph{Other cases of optimal decremental connectivity}
 There is much previous work on cases where
 decremental connectivity can be supported
 in $O(m)$ total time. Alstrup, Secher, and Spork~\cite{alstrup1997optimal} showed that decremental connectivity can be solved in optimal $O(m)$ total time on forests, answering queries in $O(1)$ time.
This was later extended to other classes of sparse graphs: planar graphs \cite{lacki2017optimal}, and
minor-free graphs \cite{holm2018good}. All
these special graph classes are sparse
with $m=O(n)$ edges. 

For general graphs, we only have the previously mentioned work
by Thorup~\cite{thorup1999decremental}, yielding a total running time of $O(m)$ for very dense graphs with $m=\Omega(n^2)$ edges. We now obtain the same linear time bound for all non-sparse graphs with $m=\Omega(n \poly \log n)$ edges.

\subsection{General reduction for $c$-edge-connectivity}\label{sec:reduction-result}
Our algorithm for decremental connectivity is based on a
general randomized reduction from decremental $c$-edge-connectivity (assuming all $m$ edges are deleted) to fully-dynamic $c$-edge-connectivity on a sparse graph
with $\tilde O(cn)$ updates. The reduction has a polylogarithmic
cost per node as well as a constant cost per edge. The previous
decremental connectivity algorithm of Thorup \cite{thorup1999decremental}
was also based on such a general reduction, but the cost per edge was
$O(\log (n^2/m))$ which is $O(1)$ only for very dense graphs with
$m=\Omega(n^2)$. Below we will describe the format of the reductions
in more detail. 

Because there are different notions of $c$-edge-connectivity, we first need
to clarify our definitions. We say that two vertices $u,v$ are
$c$-edge-connected iff there exist $c$ edge-disjoint paths between $u$
and $v$ in $G$.  It is known that $c$-edge-connectivity is an
equivalence relation; we call its classes the \emph{$c$-edge-connected
  classes}. However, for $c\geq 3$, a $c$-edge-connected class may induce a subgraph of $G$ which is not connected, so it also makes sense to consider
\emph{$c$-edge-connected components}, i.e., the maximal $c$-edge-connected
induced subgraphs of $G$.\footnote{There is no consensus in the
  literature on the terminology relating to $c$-edge-connected
  components and classes. Some authors~(e.g., \cite{GalilI93,
    GiammarresiI96}) reserve the term $c$-edge-connected components
  for what we in this paper call $c$-edge-connected classes.} It is important to note that the $c$-edge-connected components and the $c$-edge-connected classes are \emph{uniquely defined} and both induce a natural partition of the vertices of the underlying graph. Moreover, each $c$-edge-connected component of $G$ is a subset of some $c$-edge-connected class of $G$. For
$c=1,2$, the $c$-edge-connected classes are $c$-edge-connected, so the two 
notions coincide. To illustrate the difference, let us fix $c\geq 3$ and consider a graph with $c+2$ vertices $v_s, v_t, v_1, \dots, v_c$ and edges $\{v_s,v_t\}\times \{v_1, \dots, v_c\}$; while all $c$-edge-connected components in this graph are singletons, there is one $c$-edge-connected class, which is not a singleton, namely $\{v_s, v_t\}$.

We define a \emph{$c$-certificate} of $G$ to be a subgraph $H$ of $G$ that contains all edges not in $c$-edge-connected components, and a
$c$-edge-connected subgraph of each $c$-edge-connected component. 
Both Thorup's and our reduction maintains a $c$-certificate $H$ of $G$. Then, for any
$c'\leq c$, we have that the $c'$-edge-connected equivalence classes
and the $c'$-edge-connected components are the same in $G$ and $H$. As
the edges from $G$ are deleted, we maintain a $c$-certificate
with $\Ot(cn)$ edges undergoing only $\Ot(cn)$ edge insertions and deletions in total.

The uniquely defined $c$-edge-connected components of a graph may be found by
repeatedly removing cuts of size at most $c-1$. For the reductions, we
need algorithms that can help us in this process. We therefore define the
\emph{fully dynamic $c$-edge-cut problem} as follows. Suppose a graph $G$ is subject to edge insertions and/or deletions.
Then, a fully dynamic $c$-edge-cut data structure should
maintain any edge $e$ that belongs to some cut of size less than $c$. A typical application of such a data structure is to repeatedly remove such edges $e$ belonging to cuts of size less than $c$, which splits $G$
into its $c$-edge-connected components.
For each $c\geq 1$, denote by $T_c(n)$ the amortized time needed by the data structure to find an edge belonging to a cut of size less than $c$.
For example, for $c=1$ we have $T_1(n)=O(1)$ since we do not have to maintain anything. For $c=2$, the data structure is required to maintain some \emph{bridge}
of $G$ and it is known that $T_2(n)=O((\log{n}\cdot \log\log{n})^2)$~\cite{HolmRT18}. For $c\geq 3$, in turn, we
have $T_c(n)=O(n^{1/2}\poly{(c)})$~\cite{Thorup07}.

Given a fully dynamic $c$-edge-cut data structure, whose update
time for a graph on $n$ nodes is $T_c(n)$, Thorup's \cite{thorup1999decremental}
reduction maintains, in $O(m\log(n^2/m))+\Ot(c\cdot n\cdot T_c(n))$
total time, a $c$-certificate $H$ of the decremental graph $G$
starting with $n$ nodes and $m$ edges. The certificate undergoes only
$\Ot(cn)$ edge insertions and deletions throughout any sequence of
deletions issued to~$G$. We reduce here the total time to
$O(m)+\Ot(c\cdot n\cdot T_c(n))$.

Combining our reduction with the polylogarithmic fully-dynamic connectivity and \linebreak 2-edge-connectivity
algorithm of Holm, de Lichtenberg, and Thorup \cite{holm2001poly}, we can now solve
decremental connectivity and 2-edge-connectivity in $O(m)+\Ot(n)$ time. 

We can also apply
the fully dynamic min-cut algorithm of Thorup~\cite{Thorup07} which identifies
cuts of size $n^{o(1)}$ in $n^{1/2+o(1)}$ worst-case time. For $c=n^{o(1)}$,
we then maintain a $c$-certificate $H$ in $O(m+n^{3/2+o(1)})$ total time.
This includes telling which vertices are in the same $c$-edge-connected
component. If we further want to answer queries about $c$-edge-connectivity between pairs of nodes, we can apply the fully-dynamic data structure of Jin and
Sun~\cite{jin2020} to the $c$-certificate $H$. By definition,
the answers to these queries are the same in $H$ and $G$, and the algorithm
takes $n^{o(1)}$ time per update or query. Hence the total time for the
updates remains $O(m+n^{3/2+o(1)})$, and we can tell
if two vertices are $c$-edge-connected in $n^{o(1)}$ time.

\subsubsection{Results}

We will now give a more precise description of our reduction, including the log-factors  hidden in the $\Ot(cn)$ bound. Let the \emph{decremental $c$-certificate} problem be that of maintaining a $c$-certificate of $G$ when $G$ is subject to edge deletions. Recall that $T_c(n)$ denotes the amortized update time of a fully-dynamic $c$-edge-cut data structure. Thorup~\cite{thorup1999decremental} showed the following.

\begin{theorem}[Thorup~\cite{thorup1999decremental}]\label{t:thorup}
There exists a Las Vegas randomized algorithm for the decremental $c$-certificate problem with expected total update time $O(m\log{(n^2/m)}+n (c+\log{n})\cdot T_c(n)\log^2{n})$. The maintained certificate undergoes 
$O(n\cdot (c+\log{n}))$ expected edge insertions and deletions throughout, assuming $\Theta(m)$ random bits are provided. These bounds similarly hold with high probability.
\end{theorem}

In particular the total update time is $O(m)$ for very dense graphs with $\Omega(n^2)$ vertices. Our main result, which we state below, shows that an amortized constant update time can be obtained as long as the initial graph has $\Omega(n \polylog(n))$ edges.

\begin{restatable}{theorem}{thmcertificate}\label{thm-certificate}
There exists a Monte Carlo randomized algorithm for the decremental $c$-certificate problem with total update time $O(m+n(c+\log{n})\cdot T_c(n)\log^3{n}+nc\log^7{n})$.
The maintained certificate undergoes
$O(nc\log^4{n})$
edge insertions and deletions throughout. The algorithm is correct with high probability. Within this time bound, the algorithm offers a final self-check after processing all updates.
\end{restatable}

In fact, our algorithm is itself a reduction to $O(\log{n})$ instances of the decremental $c$-certificate problem on a subgraph of $G$ with $O(m/\log^2{n})$ edges. To handle these instances,
we use the state-of-the-art data structure (Theorem~\ref{t:thorup}) which costs only $O(m)$ in terms of $m$. As a result, our improved reduction  (Theorem~\ref{thm-certificate}) requires $\Theta(m)$ random bits to hold.

However, our new randomized $c$-certificate that is the key to obtaining the new reduction requires only pairwise independent sampling to work. This is in sharp contrast with the certificate of Karger~\cite{Karger99}, used in the construction of Thorup's data structure (Theorem~\ref{t:thorup}), which requires full independence, i.e., $\Theta(m)$ random bits. We show that we may instead plug our new certificate into Thorup's data structure at the cost of a single additional logarithmic factor in the running time. Since Karger's certificate constitutes the only use of randomness in Thorup's data structure, and full independence in our construction is required only for invoking Theorem~\ref{t:thorup}, we obtain the below low-randomness version of our main result.

\begin{restatable}{theorem}{thmcertificatelow}\label{thm-certificate-low}
There exists a Monte Carlo randomized algorithm for the decremental $c$-certificate problem with total update time $O(m+nc\cdot T_c(n)\log^4{n}+nc\log^7{n})$.
The maintained certificate undergoes
$O(nc\log^4{n})$
edge insertions and deletions throughout. The algorithm is correct with high probability
if $O(\polylog{n})$ random bits are provided. Within this time bound, the algorithm offers the final self-check after processing all updates.
\end{restatable}

By using Theorem~\ref{thm-certificate-low} with best known fully dynamic algorithms for different values of $c$~\cite{holm2001poly, jin2020, Thorup07}, we obtain:

\begin{restatable}{theorem}{thmfullydynamicconn}\label{thm:fully-dynamic-conn-2comm}
    There exists Monte Carlo randomized decremental connectivity and decremental 2-edge-connectivity algorithms with $O(m+n\log^7{n})$ total update time and $O(1)$ query time.
\end{restatable}
\begin{restatable}{theorem}{tdecrclasses}\label{t:decremental-c-classes}
Let $c=(\log{n})^{o(1)}$.
There exists a Monte Carlo randomized decremental $c$-edge-connectivity data structure which can answer queries to whether two vertices are in the same $c$-edge connected class in $O(n^{o(1)})$ time, and which has $O(m)+\tilde O(n^{3/2})$ total update time.
\end{restatable}
\begin{restatable}{theorem}{thmccomponents}\label{t:decremental-c-components}
Let $c=O(n^{o(1)})$.
There exists a Monte Carlo randomized decremental  $c$-edge-connected components data structure with $O(m+n^{3/2+o(1)})$ total update time and $O(1)$ query time.
\end{restatable}
All the above applications of our main result work using only $O(\polylog{n})$ random bits. They moreover each have the self-check property as well.
As discussed before, our new 2-edge-connectivity data structure implies an optimal $O(m)$-time unique perfect matching algorithm for $m=\Omega(n\polylog{n})$.

\subsection{Adaptive updates and unique perfect matching}
All our time bounds are amortized. Amortized time bounds are
particularly relevant for dynamic data structures used inside
algorithms solving problems for static graphs. In such contexts, future updates
often depend on answers to previous queries, and therefore we need algorithms that work with adaptive updates.

Our reduction works against adaptive updates as long as queries do not
reveal any details about the $c$-certificate $H$. The reduction will safely
maintain the following public information about the
$c$-edge-connected components of $G$: between deletions, each such
component will have an ID stored with all its vertices, so two
vertices are in the same $c$-edge-connected component if and only if
they have the same component ID. With the component ID, we store its size and a list with its vertices in order of increasing vertex ID.
Finally, we can have a list of all edges that are not in $c$-edge-connected
components. After each update, we can also tell what are the IDs of the
new $c$-edge-connected components, and what are the edges between them.

For the case of 2-edge-connectivity, the above means that we can
maintain the bridges of a decremental graph and we can also maintain
the connected components and their sizes without revealing what the current randomized certificate looks like. All this is needed for the unique
perfect matching algorithm of Gabow, Kaplan, and Tarhan~\cite{GabowKT01}. The
algorithm is an extremely simple recursion based on the fact that
a graph with a unique perfect matching has a bridge and all components have
even sizes. The algorithm first asks for a bridge $(u,v)$ of some component. If there
is none, the is no unique matching. Otherwise we remove $(u, v)$ and check the
sizes of the components of $u$ and $v$. If they are odd, $(u,v)$ is in the unique matching, and we remove all other incident edges. Otherwise $(u,v)$ is not in the
unique matching. The important thing here is that the bridges do not tell us anything about our 2-certificate of the 2-edge-connected components.

Thus we solve
the static problem of deciding if a graph has a unique perfect matching
in $O(m)+\Ot(n)$ time. If the self-verification reports a possible mistake,
we simply rerun. Thus we get a Las Vegas algorithm that terminates in
$O(m)+\Ot(n)$ time with high probability.

\subsection{Techniques}
Our main technical contribution is a new construction of a sparse randomized $c$-certificate that
witnesses the $c$-edge-connected components of $G$ and can be maintained under edge deletions in $G$.
In the static case, deterministic certificates of this kind have been known for decades~\cite{NagamochiI92}. However, they are not very robust in the decremental setting, where an adversary can constantly remove its edges forcing it to update frequently.
Consequently, Thorup~\cite{thorup1999decremental} used a randomized
sample-based certificate to obtain his reduction. The general idea behind this approach is to ensure that the certificate is sparse and undergoes few updates. Ideally, the sparse certificate will only have to be updated whenever an edge from the certificate is deleted. Using a fully dynamic data structure on the certificate, we may obtain efficient algorithms provided that we don't spend too much time on maintaining the certificate. Thorup's reduction had an additive
overhead $O(m\log(n^2/m))$ for maintaining the certificate, which we will reduce to $O(m)$. We shall, in fact, use Thorup's reduction as a subroutine, called on $O(\log n)$ decremental subproblems each starting with
$O(m/\log^2 n)$ edges. 

\subsubsection{Thorup's construction~\cite{thorup1999decremental}}
Let us first briefly describe how Thorup's algorithm operates on certificates and highlight difficulties in improving his reduction to linear time. First of all, the $c$-certificate is constructed as follows: initially, sample edges of $G$ uniformly with probability $P\leq 1/2$, thus obtaining a subgraph~$S$. Then, compute the $c$-edge-connected
components of $S$ and form a certificate $H$ out of two parts: (1) A \emph{recursive} certificate of $S$, and (2) the subgraph~$D$ consisting of edges of $G$ connecting distinct $c$-edge-connected components of $S$.

As proved by Karger~\cite{Karger99}, $D$ has size $\Ot(cn/P)$ with high probability. Thorup~\cite{thorup1999decremental} generalizes this by proving that $D$ undergoes only $\Ot(cn/P)$ insertions
throughout any sequence of edge deletions to $S$.
Since $D$ depends only on the $c$-edge-connected components of $S$, it is enough to have a $c$-certificate of $S$ in order to define $D$. Hence, a $c$-certificate
of $S$ (which is a graph a size $O(mP)$, i.e., a constant factor smaller) is maintained under edge deletions recursively. The recursion stops when the size
of the input graph is $O(cn)$.
To maintain $D$ at each recursive level, we first need to maintain
the $c$-edge-connected components of the (recursive) certificate of $S$ under edge deletions.
The certificate of $S$ can be (inductively) seen to have $\Ot(cn/P)$ edges and
undergo $\Ot(cn/P)$ updates. As a result, for $P =1/2$ maintaining its $c$-edge-connected components costs
$\Ot(cn\cdot T_c(n))$ total time using the fully-dynamic $c$-edge-cut data structure.
Since at each recursion level the certificate size decreases geometrically, the expected cost of all
the dynamic $c$-edge-cut data structures is $\Ot(cn\cdot T_c(n))$.

The additional cost of $O(m\log(n^2/m))$ comes from the fact that,
at each level of the recursion, when a $c$-edge-connected component in $S$ splits into two components as a result of an edge deletion, we need to find edges of $G$ between these two components in order to update $D$.
This takes $O(m\log{(n^2/m)})$ total time throughout using a standard technique
of iterating through the edges incident to the nodes in the smaller component every time a split happens~\cite{EvenS81}. The $O(\log{(n^2/m)})$ (instead of $O(\log{(n)})$) cost comes by noticing that a vertex can at most have $q$ neighbors in a component of order $q$, and that after we go through the edges of a vertex $i$ times it is in a component of order $\leq n/2^{i}$; hence it is only the first $O(\log (n/\deg(v)))$ times that all neighbors of $v$ have to be considered, so the total time spent on this becomes $O\left(\sum_{v\in V} \deg(v) \log(n/\deg(v)) \right)=O(m\log{(n^2/m)})$.

It turns out very challenging to get rid of the $O(m\log(n^2/m))$ term associated with finding cuts when components split in Thorup's reduction. If we knew that all of these cuts were \emph{small}, say of size at most $\delta$, then we could apply a whole bag of tricks for efficiently finding them in a total time of $\tilde O(n\delta)$, e.g., using invertible Bloom lookup tables \cite{goodrich2011invertible}, or the XOR-trick~\cite{DBLP:conf/soda/AhnGM12,DBLP:conf/pods/AhnGM12, kapron2013dynamic}. Unfortunately, the bound of $\tilde O(cn/P)$ only gives an average bound on the number of edges between pairs of components, and in fact there can be pairs of components having as many as $\Omega(n^{1/3})$ edges between them, as we will later show. In order to resolve this, we have to introduce a new type of sample based $c$-edge certificate obtained by only removing cuts of size at most $\delta=O(c\polylog{n})$ from $G$.
In the following three subsections, we describe the ideas behind this new certificate, the technical challenges encountered in efficiently maintaining it, and why such a certificate is relevant for decremental connectivity algorithms.

\subsubsection{A small cut sample certificate}\label{sec:certificate-high-level}
On the highest level, our $c$-edge certificate of a graph $G=(V,E)$ is constructed starting with a sample $S\subset G$. For now we assume that the edges of $S$ are sampled independently with some probability $P\leq 1/2$, but we will later see how to reduce the number of random bits needed for the sampling to $O(\polylog{n})$. 
In our algorithms, $G$ will be the current decremental graph and the sample $S$ will be made from the original graph. Thus, when $G$ has undergone a sequence of deletions, the current sample will be $S\cap G$. At any point in time, the maintained $c$-edge certificate only depends on the sample $S$ and the current graph $G$, not on the sequence of edge updates made to $G$ so far. We may therefore describe the $c$-edge certificate statically. 

The critical idea behind our certificate is to introduce a small-cut-parameter $\delta$. 
Our certificate is obtained by iteratively removing certain cuts from $G$ where each cut is allowed to be of size at most~$\delta$. We denote by $D\subset G$ the graph whose edge set consists of the edges removed in this process. The overall goal is to define this cut removal process in a way so that (1) each connected component of $G\setminus D$ is $c$-edge-connected in $S$, and (2) it is easy to detect new small cuts under edge deletions issued to $G$. We then use $S\cup D$ as our $c$-edge connectivity certificate of $G$. Importantly, we want $\delta$ to be \emph{as small as possible}, ideally $\delta= O(c\polylog(n))$. This is because $\tilde O(\delta n)$ will show up as an additive cost in our algorithm for maintaining the certificate. 
We will describe shortly how this type of certificate can be used in the design of efficient decremental $c$-edge-connectivity algorithms, but let us first demonstrate that the existence of such a cut removal process for a small~$\delta$ is non-trivial.

First of all, we could simply remove \emph{all} cuts from $G$ of size at most $\delta$ leaving us with the $(\delta+1)$-connected components. Karger's result~\cite{Karger99} implies that with $\delta=O((c+\log n)/P)$ sufficiently large, these components will remain $c$-edge connected in $S$. However, in order to maintain the small cuts, we would need a decremental $\delta$-edge connectivity algorithm. As $\delta>c$, this approach  simply reduces our problem to a much harder one. 

Suppose on the other hand that we attempted to use Thorup's sampling certificate~\cite{thorup1999decremental} described above. To simplify the exhibition, let's assume that $P=1/2$ and $c=1$. If $D$ is the set of edges between connected components of $S$, $D\cup S$ is a certificate. Thorup's algorithm recurses on $S$ to find a final certificate of $G$. At first sight it may seem like $D$ can be constructed by iteratively removing cuts of size at most $\delta =O(\log n)$ between the connected components of $S$.
After all, isn't it unlikely that a connected component of $S$ has more than, say, $100 \log n$ unsampled outgoing edges when the sampling probability is $P=1/2$? As tempting as this logic may be, it is flawed. To illustrate this, let $G=G(n,2/n)$ be the Erdős–Rényi graph obtained by sampling each edge of the complete graph on $n$ vertices independently with probability $2/n$. Let further $S$ be the subgraph of $G$ obtained by sampling each edge with probability $1/2$. Then $S$ is distributed as the Erdős–Rényi graph $G(n,1/n)$ and basic phase transition results~\cite{Erdos:1960} show that the two largest components of $S$, $C_1$ and $C_2$, almost surely have size $\Theta(n^{2/3})$. Now, we can conversely construct $G$ from $S$ by including each non-sampled edge of the complete graph with probability $\frac{1}{n-1}$, and then we expect $G$ to contain as many as $\Theta(n^{1/3})$ edges between $C_1$ and $C_2$. At some point in the iterative process, we are thus forced to remove a cut of size $\Omega(n^{1/3})$ splitting $C_1$ and $C_2$, and we would have to choose $\delta$ of at least this size (but it is possible that other examples could show that $\delta$ would have to be even larger). Our algorithms spend total time $\tilde O(n\delta)$ on finding these cuts, and if $\delta=\Omega(n^{1/3})$, this is not good enough for a linear time algorithm for non-sparse graphs.
We remark that in this example, each vertex of $G$ has degree $O(\log n)$ with high probability. Therefore, an alternative approach yielding cuts of size $O(\log n)$ would be to cut out one vertex at a time moving all incident edges to $D$. In particular this would cut the large sampled components $C_1$ and $C_2$ into singletons, one vertex at a time. Clearly, we cannot proceed like this for general graphs which may have many vertices of large degree. Nevertheless, this simple idea will be critically used in our construction.

Our actual certificate uses $\delta =O(\frac{c\log n}{P})$. The certificate can be constructed for any $P\leq 1/2$, but in our applications, $P=1/\polylog{n}$. We proceed to sketch the construction now and refer the reader to Section~\ref{sec:certificate} for the precise details. On a high level, the sample $S$ is partitioned into $\ell=O(\log n)$ samples $S_1,\dots, S_\ell$, each having size approximately $Pm/\ell$. To construct our certificate, we start by iteratively pruning~$G$ of the edges incident to vertices of degree less than $\delta$, moving these edges to $D$. The graph left after the pruning $G_1=G\setminus D$ satisfies that each vertex of positive degree has degree at least $\delta$. Next, $S_1$ defines a sample of $G_1$, $H_1=S_1\cap G_1$ and the bound on the minimum positive degree in $G_1$ guarantees that with constant probability a fraction of $3/4$ of the vertices with positive degree in the sampled subgraph $H_1$ has degree $\geq 4c$ . We prove a combinatorial lemma stating that such a graph can have at most
$5n/6$ $c$-edge-connected components.
As a result, if we contract the $c$-edge-connected components of $H_1$ in the pruned graph $G_1$, the resulting graph $G_1'$ would have at most $5n/6$ vertices.
Finally, we construct a $c$-certificate for $G_1'$ recursively using the samples $S_2,S_3,\dots$, stopping when the contracted graph has no edges between the contracted vertices (here $G$ played the role of $G_0'$). The constant factor decay in the number of components ensures that we are done after $\ell=O(\log n)$ steps with high probability.
All edges of $D$ are obtained as the removed edges of cuts of $G$ of size less than $\delta$, so $D$ will have size $O(n\delta)$. Our certificate will simply be $S\cup D$ which we prove is in fact a $c$-certificate. 

With this, we have thus completed the goal of obtaining a small cut sample certificate with~$\delta$ as small as $O(\frac{c\log n}{P})$. Abstractly, our certificate has a quite simple description: we alternate between sampling, removing small cuts around $c$-edge-connected components in the sample, and finally contracting these components.
However, in our implementation, we cannot afford to perform the contractions as described above explicitly. This makes it difficult to efficiently find the edges leaving the $c$-edge connected components of $H$ (at a given recursive level), and we must be able to do this since these $c$-edge-connected components may split as edges are deleted from $G$. It turns out that since we are only concerned with cuts of size at most $\delta$, we can in fact identify these cuts in total time $O(m)+\tilde O(\delta n)$. We will describe this in the following section.

A final property of our new randomized decremental $c$-certificate algorithm is that it requires only $O(\log^2 n)$ random bits to yield high-probability correctness bounds. This is in sharp contrast with Thorup's algorithm~\cite{thorup1999decremental} which requires $\Omega(m)$ random bits. 
On a high level, the reason we can do with few random bits is that in each step of the construction of our certificate, we only need the bounds on the number of contracted components to hold with constant probability. Indeed, we will still only have $O(\log n)$ recursive levels with high probability. This means that for the probability bounds within a single recursive level, it suffices to apply Chebyshev's inequality.
While the reduction of the number of required random bits is nice, the main point, however, is that with our new certificate we can get down to constant amortized update time per edge-deletion for decremental (2-edge)-connectivity for all but the sparsest graphs.

\subsubsection{Maintaining our certificate}
As edges are deleted from $G$, the recursive structure of the $c$-certificate $H$ changes. Indeed, the deletion of an edge may cause the following changes in one of the recursive layers of $H$: (1) introduce a cut of size less than $\delta$ surrounding a $c$-edge-connected component or (2) break a $c$-edge-connected component in two. In the first case, the edges of the cut have to be moved to $D$, and deleted from the current and later layers of $H$, causing further cascading.
When a $c$-edge-connected component (in a recursive layer) of $H$ breaks in two, we need to determine whether either of the new components has less than $\delta$ outgoing edges in $G$. If we use the standard technique of iterating over all the edges incident to nodes of the smaller component, this again incurs an $O(\log(n^2/m))$ cost per edge which is insufficient.
However, as we only care about components with at most $\delta$ outgoing edges, it turns out that we can do better.
We define the \emph{boundary} of a component~$C$ of some graph $H\subset G$ to be the set of edges of $G$ with one endpoint in $C$, and another in $V\setminus C$.
To overcome the $O(\log{(n^2/m)})$ cost per edge, we prove that we can maintain boundaries of size at most $\delta$ under splits using a Monte Carlo randomized algorithm in ${O(m+n\delta\polylog{n})}$ total time.  We realize this result by deploying the XOR-trick~\cite{DBLP:conf/soda/AhnGM12,DBLP:conf/pods/AhnGM12, kapron2013dynamic} in a somehow unusual manner. In particular, we randomly partition the set of edges of $G$ and use the XOR-trick to detect the parts that are relevant for scanning, as opposed to standard applications of the XOR-trick, which are used to detect a single replacement edge over a polylogarithmic number of independent samples which unavoidably introduces a polylogarithmic dependence in the cost per edge. Since we are not interested in discovering a single edge incident to some set, we only need to apply the XOR-trick once, which allows us to keep the running time linear in $m$.

\subsubsection{Combining our certificate with Thorup's algorithm}
With the certificate as above, the overall idea for a decremental connectivity algorithm is to maintain a $c$-certificate of (each recursive layer of) the \emph{decremental} graph $H=S\setminus D$ using the algorithm by Thorup~\cite{thorup1999decremental}. 
By choosing $P=1/\log^2 n$, $S$ has $m'=O(m/\log^2 n)$ edges with high probability, so employing the algorithm of Theorem~\ref{t:thorup} on each recursive layer takes total time $O(m'\log^2 n+ncT_c(n)\polylog{n} )=O(m+ncT_c(n)\polylog{n} )$ with high probability. Let $H^*$ be the $c$-certificate thus obtained for $H$. Using a fully dynamic $c$-edge-connectivity algorithm on $H^* \cup D$ (which undergoes $O(cn \polylog{n})$ updates), we maintain a $c$-edge certificate of $G$. As $H^* \cup D$ undergoes $O(cn \polylog{n})$ updates, running the fully dynamic algorithm takes total time $O(cnT_c(n) \polylog{n})$.

We remark that for $c=1,2$ we could instead use a fully dynamic $c$-edge connectivity algorithm on $H$ with polylogaritmic update and query time at the price of a smaller $P$ (which would incur more log-factors in our final time bound). For $c>2$, however, we only know that $T_c(n)=O(n^{1/2}\poly(c))$. Since, running a fully dynamic algorithm on $H$ takes total time $\Omega(mT_c(n)/\polylog{n})$, this is insufficient to obtain linear time algorithms for dense graphs.

\subsubsection{Final self-check}
Let us finally describe the ideas behind the final self-checks claimed in Theorem~\ref{thm-certificate} and~\ref{thm-certificate-low} in a more general context.
In particular, we show that if a randomized Monte Carlo dynamic algorithm satisfies some generic conditions then it can be augmented to detect, at the end of its execution, whether there is any chance that it answered any query incorrectly. That is, if the self-check passes then it is guaranteed that all queries were answered correctly throughout the execution of the algorithm. Otherwise, it indicates that some queries might have been answered incorrectly. The self-check property is particularly useful in applications of dynamic algorithms as subroutines in algorithms solving static problems, that is, it enables static algorithms to exhibit Las Vegas guarantees instead of the Monte Carlo guarantees provided by the dynamic algorithm, as they can simply re-run the static algorithm with fresh randomness until the self-check passes. 

The properties of a dynamic algorithm amenable to a self-check behavior are as follows:
\begin{itemize}
    \item Once a mistake is made by the dynamic algorithm it should be detectable and any subsequent updates of the algorithm do not correct the mistake before it is detected.
    \item If the dynamic algorithm is stopped at any point in time, it should be able to still perform the self-check within the guaranteed running time of the algorithm.
\end{itemize}

In our algorithm, as long as the $c$-certificate maintained by our algorithm is correct, the $c$-edge-connectivity queries answered by our algorithm exhibit the same guarantees as the fully dynamic $c$-edge-connectivity algorithm running on the $c$-certificate $H$. Hence, we only need to detect potential mistakes in the process of maintaining the $c$-certificate $H$. Such mistakes only happen with probability $n^{-\Omega(1)}$. 

The $c$-certificate $H\subseteq G$ of $G$ that we maintain is such that it includes every edge
from $G$ that is not in a $c$-edge-connected component of $G$ (more specifically, not in a $c$-edge-connected component of the maintained sample $S\subset G$, which includes all edges not in $c$-edge-connected components of~$G$) and such that
every $c$-edge-connected component of $G$ is also a $c$-edge-connected component
in~$H$. A different equivalent formulation is that for every edge
$(u,v)$ of $G$, if $(u,v)$ is not in $H$, then $u$ and~$v$ are in the
same $c$-edge-connected component of $H$.

As we later show, the only possible error in the maintenance of the $c$-certificate is that an edge is missing from $H$ while its endpoints are not $c$-edge-connected in $S$. Therefore, to satisfy the properties above, any incorrectly missing edge in the $c$-certificate should remain missing until the error is detected. 
We show (in the proof of Theorem~\ref{thm-certificate}) that before each update to $H$ we can check whether the edge should have been part of the certificate but was omitted due to an error.
It is trivial to check for such mistakes during the execution of our algorithm, as follows. For any edge that is deleted from the graph we simply check whether the two endpoints of the deleted edge belong to distinct $c$-edge-connected components of $H$ and, if so, declare a potential error. For every edge that is inserted to the certificate, due to the updates following an edge deletion, we check whether the edge was supposed to be part of the certificate before the last edge deletion but was omitted due to an error.
In either case, our self-check flags an execution invalid only if there has been a mistake (which might or might not have affected $c$-edge-connectivity queries on the certificate), which happens with low probability.
As long as each vertex knows the ID of its $c$-edge-connected component in the $c$-certificate, the aforementioned check takes constant time to perform per edge deletion, and thus does not affect the overall running time of our algorithm.
Since our algorithm is Monte Carlo randomized and we only flag an execution invalid if a mistake in the maintenance of the certificate was detected, any single execution of our algorithm is flagged invalid with probability $n^{-\Omega(1)}$.

Notice that if the algorithm is terminated before all edges are deleted, we can simply iterate over the remaining edges and apply the aforementioned check for each remaining edge.

\section{Preliminaries}

The problem of dynamic connectivity consists in designing a data structure that
maintains a dynamic undirected graph and supports two operations: an
\emph{update} operation which modifies the maintained graph, and a
\emph{query} operation, which asks if two given vertices belong to the
same connected component of the current graph.

Dynamic connectivity comes in three variants, which differ in the allowed types of update operations.
The most general is the \emph{fully dynamic} connectivity problem, in which each update may either add or remove a single edge.
The two more restricted variants are \emph{incremental} connectivity -- each update adds a single edge, and \emph{decremental} connectivity -- each update removes a single edge.

\paragraph{Graphs.}
Throughout the paper we consider undirected graphs which may have parallel edges, but not self-loops.  Generally, for $G=(V,E)$, we denote by $n$ and $m$ the number of vertices and edges of $G$, respectively. When referring to other
graphs $H=(V,E')$, we write $|H|$ to denote $|E'|$. 

If $G'=(V', E')$ is a graph with $V'=V$ and $E'\subset E$ then $G'$ is a \emph{subgraph} of $G$, denoted $G'\subseteq G$.
If $W$ is a set of edges, then we denote by $G\setminus W$ the graph with vertex set $V$ and edge set $E'=E\setminus W$. We often use $G\setminus H$ to denote $G\setminus E(H)$. Finally, if $G=(V, E_G)$ and $H=(V, E_H)$ are graphs on the same set of vertices, then $G\cup H$ is the graph with vertex set $V$ and edge set $E_G\cup E_H$.

	Let $A, B\subset V$ be subsets of vertices of $G=(V, E)$. The set of edges from $A$ to $B$ in $G$ is denoted $E_G(A, B)$.
	For $A\subset V$, we denote by $\partial_G(A)=E_G(A, V\setminus A)$ the \emph{boundary} of $A$ in $G$.
\vspace{-3mm}
\paragraph{Cuts.}
    A \emph{cut} of $G=(V,E)$ is a partition of the vertices of $G$ into two non-empty sets $V_1, V_2$. We shall mostly identify the cut with the set of edges crossing the cut, $E(V_1, V_2)$. The number of such edges is the \emph{size} of the cut. 
    A cut of $G$ of size $c\in \N$ is called a \emph{$c$-cut} of $G$.
    A \emph{simple cut} of $G$ is a set of edges $S\subset E$ such that $G\setminus S$ has exactly one connected component more than $G$ and such that for every proper subset $S'\subsetneq S$, $G\setminus S'$ has the same connected components as $G$.

\paragraph{Edge connectivity.}
Let $c$ a be positive integer. Two distinct vertices $u, v$ of $G$ are \emph{$c$-edge-connected} if there exist $c$ pairwise edge-disjoint paths from $u$ to $v$. Being $c$-edge-connected is an equivalence relation on the vertices of $G$ and we call the corresponding equivalence classes the \emph{$c$-edge-connected classes}. The graph $G$ is \emph{$c$-edge-connected} if it contains no cut of size $<c$. Equivalently, a graph $G$ is $c$-edge-connected if every pair of distinct vertices of $G$ are $c$-edge-connected. It is worth noting that the graph consisting of a single vertex is $c$-edge-connected since it contains no cut.
The \emph{$c$-edge-connected components} of $G$ are the maximal induced $c$-edge-connected subgraphs of $G$, i.e., an induced subgraph $C$ of $G$ is a $c$-edge-connected component if it is $c$-edge-connected and there exists no intermediate subgraph $C\subsetneq C' \subset G$ such that $C'$ is $c$-edge-connected.
	
	For $c=1$, we have simpler terminology. We say that 1-edge-connected vertices are \emph{connected} and call the 1-edge-connected components of $G$  the \emph{connected components} or simply \emph{components} of~$G$.
	We also use $\mathcal{C}(G)$ to denote the set of connected components of $G$.

\paragraph{Fully dynamic $c$-edge-cut.}
The crucial ingredient in obtaining our general decremental algorithm
for arbitrary $c\geq 1$ is the \emph{fully-dynamic $c$-edge-cut} data structure, defined
as follows. Let $G$ be a graph. Then, the data structure maintains
any edge $e$ (if one exists) satisfying the following: $e$ belongs to some cut of size $<c$ of the connected component of $G$ containing $e$. For each $c\geq 1$, we denote by $T_c(n)$ the amortized update time bound of such a data structure
that holds whp.

For example, for $c=1$ we have $T_1(n)=O(1)$ since we do not have to maintain anything. For $c=2$, the data structure is required to maintain some \emph{bridge}
of $G$ and it is known that $T_2(n)=O((\log{n}\cdot \log\log{n})^2)$~\cite{HolmRT18}. For $c\geq 3$, in turn, we
have $T_c(n)=O(n^{1/2}\poly{(c)})$~\cite{Thorup07}.

\paragraph{Chernoff Bound.}
In our analysis we will occasionally need the classic Chernoff concentration bounds. We state a version here for convenience. 
\begin{theorem}[Chernoff Bound]\label{thm:chernoff}
	Let $X_1, \dots, X_n$ be independent random variables supported on $[0, 1]$ and denote by $\mu=\E{\sum_{i=1}^nX_i}$ the mean of their sum. For every $\delta>0$,
	\begin{align*}
		\forall \mu'\geq \mu\colon \Pr\left[ \sum_{i=1}^nX_i>(1+\delta) \mu' \right]&\leq e^{-\frac{\mu'\delta^2}{2+\delta}},\hspace{10mm}
		\forall \mu'\leq \mu\colon \Pr\left[ \sum_{i=1}^nX_i<(1-\delta) \mu' \right]&\leq e^{-\frac{\mu'\delta^2}{2}}.
	\end{align*}
\end{theorem}

\paragraph{Uniform edge sampling.}
Finally, for a graph $G=(V, E)$ and $p\in [0, 1]$ a real number, we define the \emph{uniform edge sampling} $G(p)$ as follows. Let $\{X_e\}_{e\in E}$ be independent Bernoulli random variables with parameter $p$ and $E'=\{e\in E\mid X_e=1\}$. Then $G(p)=(V, E')$.

\section{Some Useful Properties of $c$-Edge-Connected Components}\label{sec:useful-properties}
In this section, we present structural lemmas regarding the $c$-edge-connected components, in particular in graphs where a significant fraction of the vertices has degree at least $\Omega(c)$. As described in Section \ref{sec:certificate-high-level}, at each level of the recursive construction of our $c$-certificate, all $c$-edge-connected components of the sampled graph are contracted into single vertices in the recursive calls to the next levels. The following structural results show that the number of vertices between subsequent levels shrink by a constant factor, implying a $O(\log n)$ bound on the number of levels in our certificate.
We note that some proofs from Sections \ref{sec:useful-properties}, \ref{sec:supporting}, and \ref{sec:decremental-certificate} can be found in Section \ref{sec:omitted-proofs}.
\begin{restatable}[Bencz\'ur and Karger \cite{Benczur15}]{lemma}{benczur} \label{lem:many_edges_implies_con_comp}
	Let $c$ and $n$ be positive integers. Every graph on $n$ vertices with strictly more than $(c-1)(n-1)$ edges contains a non-trivial $c$-edge-connected component.
\end{restatable}

\begin{restatable}{corollary}{ebce}\label{cor:edges_between_c_edge}
	Let $c$ be a positive integer and $G$ be a graph on $n$ vertices. Denote by $q_c$ the number of $c$-edge-connected components of $G$. Then the number of edges connecting distinct $c$-edge-connected components of $G$ is at most $(c-1)(q_c-1)$.
\end{restatable}

Another central lemma is the following, stating that if a graph on $n$ vertices has at least $3/4n$  vertices with sufficiently high degree, the number of $c$-edge-connected components is at most $5/6n$.

\begin{lemma}\label{lem:min_deg_yields_few_con_comp}
    Let $G=(V, E)$ be a graph on $n$ vertices such that
    at least $3n/4$ of its vertices have degree at least $4c$. The number of $c$-edge-connected components of $G$ is at most $5n/6$.
\end{lemma}
\begin{proof}
    Denote by $q_c$ the number of $c$-edge-connected components of $G$. Let $V_0\subset V$ be the set of vertices of $G$ that are trivial $c$-edge-connected components. Then $q_c\leq (n-|V_0|)/2 + |V_0| = (n+\abs{V_0})/2$. Furthermore, every edge incident to a vertex of $V_0$ connects distinct $c$-edge-connected components of~$G$.
    Since there are at most $n/4$ vertices with degrees less than $4c$, at least $|V_0|-n/4$ vertices in $V_0$
    have degree at least $4c$.
    Hence, there are at least $2c\cdot (|V_0|-n/4)$ edges
    incident to vertices in $V_0$.
    By \cref{cor:edges_between_c_edge}, we have
    $c\cdot (n+\abs{V_0})/2\geq 2c(\abs{V_0}-n/4)$,
    which implies $\abs{V_0}\leq 2n/3$. The conclusion follows since $q_c\leq (n+|V_0|)/2\leq 5n/6$.
\end{proof}

\section{The new $c$-certificate}\label{sec:certificate}
In this section we describe our new $c$-certificate that is instrumental in obtaining the near-optimal decremental connectivity algorithm.
A $c$-certificate $H$ of a graph $G$ allows us to answer queries about $c$-edge-connected components and $c$-edge-connected classes of $G$.

\begin{definition}[$c$-certificate]
	Let $G$ be a graph and $c\in \N$. A \emph{$c$-certificate} for $G$ is a subgraph $H\subseteq G$ such that the $c$-edge-connected components and classes of $G$ are preserved in~$H$.
\end{definition}

Let $\delta >c$ and $\ell\geq 1$ be integers to be set later. The certificate is defined based on $\ell+1$ levels of graphs $H_i,G_i\subseteq G$ for $i=0,\ldots,\ell$.

The first step is to sample graphs $H_0^0,H_1^0,\ldots,H_\ell^0$ that
constitute the basis for graphs $H_0,\ldots,H_\ell$.
Let $p\in (0,1)$ be a real number to be fixed later.
The sampled subgraphs satisfy $(V,\emptyset)=H_0^0\subseteq H_1^0\subseteq \dots\subseteq H_\ell^0$.
Specifically, for each $i=1,\ldots,\ell$, $H_i^\ell$ is constructed as follows.
Let $s=\lceil pm\rceil$ and suppose $E=\{e_1,\ldots,e_m\}$.
Let $r_i:\{1,\ldots,s\}\to \{1,\ldots,m\}$ be a \emph{pairwise independent random number generator}, or, in other words,
a pairwise independent hash function.
Then, we set:
\begin{equation*}
H_i^0:=H_{i-1}^0\cup \left(V,\left\{e_{r_i(1)},e_{r_i(2)},\ldots,e_{r_i(s)}\right\}\right).
\end{equation*}
However, we stress that for each level $i$, we require the random generator to be fully independent from the random generators at previous levels $1,\ldots,i-1$.
It is well known~\cite{CarterW79} that a pairwise independent random number generator can be implemented using $\Theta(\log{n})$ random bits
so that it generates numbers in constant time in the word RAM model. As a result, if $\Theta(\ell\log{n})$ random bits are provided, each $H_i^0$ can be constructed in $O(si)=O(mpi)$ time and has $O(mpi)$ edges.

Now, the graphs $H_0,G_0,H_1,G_1,\ldots,H_\ell,G_\ell$ are defined inductively.
Set $G_{-1}=G$. Then for $i=0,\ldots,\ell$ the graphs $H_i,G_i$ are obtained as follows.
First, the graph $H_i$ is obtained from $H_i^0\cap G_{i-1}$
by repeatedly removing all the cuts of size less than $c$. In other words, $H_i$ equals the $c$-edge-connected components of $H_i^0\cap G_{i-1}$.
Afterwards, the graph $G_i$ is in turn obtained from $G_{i-1}$ as follows. While for some $c$-edge-connected component $C$ of $H_i$, we have
$|\partial_{G_i}(C)|<\delta$, the edges of the boundary
$\partial_{G_i}(C)$ are removed from both $H_i$ and $G_i$.
Equivalently, one could obtain $G_i$ by first contracting all the $c$-edge-connected components of $H_i$ in the initial $G_i$, then repeatedly removing edges incident to vertices of degree $< \delta$ in the contracted graph, and finally undoing all the contractions.

By the construction, the graphs $H_1,\ldots,H_\ell$ and $G_1,\ldots,G_\ell$ satisfy the following properties:
\begin{enumerate}[label=(\arabic*)]
    \item Every connected component of $H_i$ is $c$-edge-connected.
    \item $H_i\subseteq G_i$ and $G_{i+1}\subseteq G_i$ for all $i\geq 0$.
    \item Each connected component $C$ of $H_i$ satisfies either $\partial_{G_i}(C)=\emptyset$ or
    $|\partial_{G_i}(C)|\geq \delta$.
\end{enumerate}
Moreover, we have the following property.
\begin{lemma}\label{l:incl}
    For any $i=0,\ldots,\ell-1$, $H_i\subseteq H_{i+1}$.
\end{lemma}
\begin{proof}
First of all, note that $H_i\subseteq H_{i+1}^0\cap G_i$
since $H_i\subseteq G_i$ and $H_i\subseteq H_i^0\subseteq H_{i+1}^0$.
Moreover, each component of $H_i$ is $c$-edge-connected
so it is contained in some $c$-edge-connected component of any supergraph of $H_i$, in particular $H_{i+1}^0\cap G_i$. As a result, when obtaining $H_{i+1}$ from $H_{i+1}^0\cap G_i$ by taking the $c$-edge-connected components, we never remove edges of $H_i$.
\end{proof}
Figure~\ref{fig:illustration_for_abstract_alg} shows the inclusion relations between the graphs $G_i,H_i$
(property~(2) and Lemma~\ref{l:incl}).

\begin{figure}[ht]
    \begin{center}
        \begin{tikzcd}[remember picture]
            H_\ell &\subset & G_\ell\\
            H_{\ell-1} &\subseteq & G_{\ell-1}\\
            \vdots &&\vdots\\
            H_{1} &\subseteq & G_{1}\\
            H_{0} &\subseteq & G_{0}\\
            &&G
        \end{tikzcd}
        \begin{tikzpicture}[overlay,remember picture]
            \path (\tikzcdmatrixname-1-1) to node[midway,sloped]{$\supseteq$}
            (\tikzcdmatrixname-2-1);
            \path (\tikzcdmatrixname-1-3) to node[midway,sloped]{$\subseteq$}
            (\tikzcdmatrixname-2-3);
            
            \path (\tikzcdmatrixname-2-1) to node[midway,sloped]{$\supseteq$}
            (\tikzcdmatrixname-3-1);
            \path (\tikzcdmatrixname-2-3) to node[midway,sloped]{$\subseteq$}
            (\tikzcdmatrixname-3-3);
            
            \path (\tikzcdmatrixname-3-1) to node[midway,sloped]{$\supseteq$}
            (\tikzcdmatrixname-4-1);
            \path (\tikzcdmatrixname-3-3) to node[midway,sloped]{$\subseteq$}
            (\tikzcdmatrixname-4-3);
            
            \path (\tikzcdmatrixname-4-1) to node[midway,sloped]{$\supseteq$}
            (\tikzcdmatrixname-5-1);
            \path (\tikzcdmatrixname-4-3) to node[midway,sloped]{$\subseteq$}
            (\tikzcdmatrixname-5-3);

            \path (\tikzcdmatrixname-5-3) to node[midway,sloped]{$\subseteq$}
            (\tikzcdmatrixname-6-3);
        \end{tikzpicture} 
    \end{center}
    \caption{Illustration for \cref{alg:certificate_sketch}}\label{fig:illustration_for_abstract_alg}
\end{figure}

\begin{lemma}\label{c-certificate}
    There exists $\ell=O(\log{n})$ such that if $p\ell<1$ and $p\delta \geq 32c$, then,
    with high probability,
    the connected components of $G_\ell$
    are $c$-edge-connected and equal
    to the connected components of $H_\ell$.
\end{lemma}
\begin{proof}
Denote by $G_i'$ the graph $G_i$ with the components of $H_i$ contracted. By property~(3), every vertex in $G_i'$ has degree either 0 or at least $\delta$.
Let $k_i$ be the number of positive (in fact, at least $\delta$) degree
vertices of $G_i'$.

Recall that $H_{i+1}$ contains precisely the edges inside the $c$-edge-connected components of ${H^0_{i+1}\cap G_i}$.
Moreover, let $H_{i+1}'\subset G_i'$ be the graph $H^0_{i+1}\cap G_i$ with the ($c$-edge-connected) components of $H_i$ contracted.
Consider some vertex $v'$ of $G_i'$.
If $v'$ has degree $0$ in $G_i'$ it does so
as well in $H_{i+1}'$.
Otherwise, by property~(3),
$v'$ has degree at least $\delta$
in $G_i'$. Recall that the edges of $H^0_{i+1}\setminus H^0_i$
are chosen independently of the graphs $G_i'$ and $H_i$ (which only depend on the randomness from levels $0,\ldots,i$)
via sampling with replacement $s=\lceil pm\rceil$ edges using a pairwise independent
random number generator $r_{i+1}$.
Consider a random variable $X_{v'}$ equal to the degree of $v'$ in $H_{i+1}'$. We now prove that $v'$ has degree less than $4c$, i.e., $X_{v'}<4c$ with probability no more than $\frac{3}{16}$.

To this end, we introduce two more random variables $Y_{v'},Z_{v'}$:
\begin{itemize}
    \item $Y_{v'}$ equals the number of times an edge incident to $v'$ is sampled when sampling $H_{i+1}^0\setminus H_i^0$:
    \begin{equation*}
        Y_{v'}=\sum_{1\leq j\leq s}[e_{r_{i+1}(j)}\text{ is incident to }v'\text{ in }G_i']
    \end{equation*}
    \item $Z_{v'}$ equals the number of collisions incident to $v'$ during sampling $H_{i+1}^0\setminus H_i^0$, i.e.,
    \begin{equation*}Z_{v'}=\sum_{1\leq j<k\leq s} [r_{i+1}(j)=r_{i+1}(k)\text{ and } e_{r_{i+1}(j)}\text{ is incident to }v'\text{ in }G_i']
    \end{equation*}
\end{itemize}
Let $d=\deg_{G_i'}(v')\geq \delta$.
Since $Y_{v'}$ is a sum of $s$ pairwise independent indicator variables with mean $d/m$, we have:
\begin{align*}
\E{Y_{v'}}&=s\cdot d/m=\lceil pm\rceil\cdot d/m\geq pd.\\
\Var{Y_{v'}}&=s\cdot (d/m)\cdot (1-d/m)\leq \lceil pm\rceil \cdot (d/m)\leq 2pm\cdot (d/m)=2pd.
\end{align*}
By $pd\geq p\delta\geq 32c\geq 32$ and Chebyshev's inequality
$\Pr[Y_{v'}\leq (1-\eps)\mu]\leq\frac{\Var{Y_{v'}}}{\eps^2(\E{Y_{v'}})^2}$ we have:
\begin{equation}\label{eq:cheb}
\Pr[Y_{v'}\leq pd/4]\leq \Pr[Y_{v'}\leq \E{Y_{v'}}/4]\leq \frac{2pd}{\frac{9}{16}(dp)^2}\leq \frac{32}{9pd}\leq \frac{1}{9}<\frac{1}{8}.
\end{equation}
By pairwise independence we also have:
\begin{equation*}
    \E{Z_{v'}}=\sum_{1<j<k\leq s}\frac{1}{m}\cdot\frac{d}{m}\leq \frac{s^2}{2}\cdot \frac{d}{m^2}\leq \frac{4p^2m^2}{2}\cdot \frac{d}{m^2}=2p^2d.
\end{equation*}
Since we are aiming at proving the lemma for $\ell=\gamma\cdot \log{n}$ for a constant $\gamma$ of our choice, we can without loss of generality require that $p\ell<1$ implies $p\leq1/256$.
Hence, using Markov's inequality we obtain:
\begin{equation*}
    \Pr[Z_{v'}\geq pd/8]\leq \frac{\E{Z_{v'}}}{pd/8}\leq 16p\leq \frac{1}{16}.
\end{equation*}
Note that we have $X_{v'}\geq Y_{v'}-Z_{v'}$.
So, $X_{v'}\leq pd/8$ implies $Y_{v'}-Z_{v'}\leq pd/8$.
This in turn implies that either $Y_{v'}\leq pd/4$ or 
$Z_{v'}\geq pd/8$. As a result, via the union bound
we get:
\begin{equation*}
    \Pr[X_{v'}\leq pd/8]\leq \Pr[Y_{v'}-Z_{v'}\leq pd/8]\leq \Pr[Y_{v'}\leq pd/4]+\Pr[Z_{v'}\geq pd/8]\leq \frac{1}{8}+\frac{1}{16}=\frac{3}{16}.
\end{equation*}
By $pd\geq 32c$ we have that $X_{v'}<4c$ implies $X_{v'}\leq pd/8$, so we obtain
$\Pr[X_{v'}<4c]\leq \frac{3}{16}$ as desired.

Now let us consider the probability $q$ that more than a fraction of 1/4 of such $n'$ vertices $v'$ (with degree at least $\delta$ in $G_i'$) have degree less than $4c$ in $H_{i+1}'$.
By~\eqref{eq:cheb}, the expected number of such vertices is clearly
no more than $\frac{3n'}{16}$.
As a result, by Markov's inequality,
    $q\leq \frac{3n'}{16}\cdot\frac{4}{n'}=\frac{3}{4}$.
In other words, with probability at least $1-q\geq 1/4$, at least a fraction of $3/4$ of positive-degree vertices $v'$ of $G_i'$ will have degree at least $4c$ in $H_{i+1}'$.

Observe that since $G_i\supseteq G_{i+1}\supseteq\ldots G_\ell\supseteq H_\ell\supseteq\ldots\supseteq H_i$,
the isolated vertices of $G_i'$ (which are obviously also isolated in  $H_{i+1}'$) correspond
to $c$-edge-connected components of $H_i$ that are also
$c$-edge-connected components of $H_{i+1},H_{i+2},\ldots,H_\ell,G_i,\ldots,G_\ell$.
Since $H_i\subseteq H_{i+1}$,
by Lemma~\ref{lem:min_deg_yields_few_con_comp},
with probability at least $1/4$,
the $k_i$ non-isolated vertices of $G_i'$ are ``merged''
into at most $5k_i/6$ $c$-edge-connected components of $H_{i+1}$. Observe that those are the only $c$-edge-connected
components of $H_{i+1}$ that can give rise to positive-degree
vertices of $G_{i+1}'$.
Consequently, with probability $\geq 1/4$ we have $k_{i+1}\leq 5k_i/6$.
This proves that $k_i$ is very likely to decrease
geometrically with $i$.
More concretely, the quantity $k_i$ is $0$ for $i=\ell=z\cdot\log{n}$ (where $z$ is a sufficiently large constant) with high probability via the Chernoff bound.

Note that the lemma follows by $k_\ell=0$, the fact that $H_\ell\subseteq G_\ell$, and
property (1) for $i=\ell$.
\end{proof}

Finally, we obtain a $c$-certificate by taking a union
of $H_\ell$ and $G\setminus G_\ell$. Roughly speaking, since $H_\ell$ sparsifies the $c$-edge-connected components of $G_\ell$, replacing the subgraph $G_\ell$
with $H_\ell$ preserves both the $c$-edge-components and $c$-edge-classes. The formal proof can be found in the Appendix.

\begin{restatable}{lemma}{hldcert}\label{l:hldcert}
     Let $D:=G\setminus G_\ell$. $H_\ell\cup D$ constitutes a $c$-certificate for $G$.
\end{restatable}

We now show that the basis of our construction, i.e., pairwise independent sampling at $O(\log{n})$ levels, yields an interesting low-randomness alternative to Karger's result~\cite{Karger99} saying that if a graph $G$ is $c'$-edge-connected graph,
where $c'=\Omega((c+\log{n})/p)$,
then $G(p)$ is $c$-edge-connected with high probability (depending on the constant hidden in the $\Omega$ notation).
Roughly speaking, Karger's proof
applies a Chernoff bound to an \emph{exponential} number of cuts in $G$ and therefore requires sampling with full independence, i.e., $\Theta(m)$ random bits.
We show that the graph $H_0^\ell$ has a similar property, but requires only \emph{polylogarithmic} number of random bits: pairwise independence requires $O(\log{n})$ bits, and there are $O(\log{n})$ sampling levels.

\begin{lemma}\label{l:karger-alt}
    Let $\ell=\Theta(\log{n})$ be as in Lemma~\ref{c-certificate}.
    Let $p'\in (0,1)$.
    Suppose $G$ is $c'$-edge-connected, where
    $c'=\Omega(c\log{n}/p')$ and the constant hidden in the $\Omega$ notation is sufficiently large.
    Then the sampled graph $H_\ell^0$, defined as before, has $O(mp')$ edges and is $c$-edge-connected with high probability.
\end{lemma}
\begin{proof}
By Lemma~\ref{c-certificate}, if $p\ell<1$ and $p\delta\geq 12c$, then
$H_\ell$ has the same $c$-edge-connected components as $G_\ell$. In particular, for any $c'\geq c$, if $G_\ell$ is $c'$-edge-connected, then $H_\ell$
is $c$-edge-connected.
Observe that $G_\ell$ can be obtained from $G$ be repeatedly removing from $G$ some cuts of size less than $\delta$. However, if $G$ is $c'$-edge-connected for $c'\geq \delta$, no such cuts exist in $G$ so in fact we have $G=G_\ell$.
Let $p=p'/\ell$, $\delta=12c/p=\Theta(c\log{n}/p')$.
It follows that if $G$ is $\geq \delta$-edge-connected,
i.e., $\Omega(c\log{n}/p')$-edge-connected, then $H_\ell$ is $c$-edge-connected with high probability.
Since $H_\ell\subseteq H_{\ell}^0$, so is $H_{\ell}^0$.
\end{proof}

\section{Decremental Maintenance of a $c$-certificate}\label{sec:decremental-certificate}

In this section we give an algorithm for maintaining a $c$-certificate of Section~\ref{sec:certificate}
for a graph $G$ that is subject to edge deletions.
Even though the graph $G$ is decremental, our maintained certificate
will undergo both edge insertions and deletions.
However, we will show that, for non-sparse graphs, it is possible
that the certificate undergoes only a sublinear-in-$m$ number
of updates throughout. Moreover, we will show that it is possible
to maintain the certificate in roughly $O(m)+{\Ot(c\cdot n\cdot T_c(n))}$ time which
is $O(m)$ for non-sparse graphs,
depending on
the known upper bounds on $T_c(n)$.

In this section we disregard the total number of random bits needed to achieve the claimed bounds. We discuss how the data structure can be implemented using only $O(\polylog{n})$ random bits later in Section~\ref{sec:bits}.

During initialization the algorithm samples $H_0^0, \ldots, H_\ell^0$ as described in Section~\ref{sec:certificate} and initially sets $H_i:=H_i^0$ for all $i$.
Furthermore, at the start of the algorithm, each $G_i$ is (conceptually) initialized to $G$. We stress at this point that the $\ell$ graphs $G_i$ are stored explicitly only in the basic version of the algorithm. The refined version avoids that, as will be discussed later on. 
The initialization of graphs $H_i$ and $G_i$ -- so that they match their definition from Section~\ref{sec:certificate} -- is completed using the update procedure as described below. 

The update procedure simply rebuilds the  subsequent levels $i=0,\ldots,\ell$ of the certificate
according to their definition from Section~\ref{sec:certificate}.
Each of these maintained graphs $H_i,G_i$ is decremental in time.
A deletion of a single edge $e$ of $G$ (or the final step of the initialization) may in general cause a deletion of a larger set $A$ of edges from
the level-$j$ graphs $H_j,G_j$.
More precisely, $e$ is first deleted from $G_0$.
This in turn may give rise to new vertices of degree
less than $\delta$ in $G_0$. Recall that edges incident to such vertices should be repeatedly removed from $G_0$ until there are none; denote by $A$ the set of edges removed in this process plus the edge $e$.
Observe that all graphs at levels $1,\ldots,\ell$ are subgraphs of $G_0$, so all the edges from $A$ should also be removed from these graphs.
More generally, if the level $j$ is passed a set $A$
of edges to be removed, these edges are first removed
from both $H_j$ and $G_j$.
As a result of this change, some new cuts of size less than $c$ may appear in $H_j$, and consequently some $c$-edge-connected components of $H_j$ may split.
The splits (as well as the deletions of edges from $A$) may give rise to new boundaries $\partial_{G_j}(C)$ of size less than $\delta$ that
have to be detected and pruned. The removed boundaries are added to the set $A$ to be passed to subsequent levels. 

Algorithm~\ref{alg:certificate_sketch} summarizes this conceptual implementation
of the above procedure for rebuilding the certificate. In the algorithm, as well as in the following we set $D:=G\setminus G_\ell$.

\begin{algorithm}[h!]
    \SetKwInOut{Input}{Input}
    \SetKwInOut{Parameters}{Parameters}
    \SetKwInOut{Maintains}{Maintains}
    
    \SetKwProg{Fn}{Function}{:}{}
    \SetKwProg{Proc}{Procedure}{:}{}
    \SetKwProg{PFn}{Private Function}{:}{}
    \SetKwFunction{FInit}{Initialize}
    \SetKwFunction{FDel}{Delete}
    \SetKwFunction{FClean}{CleanUp}
    
    \Input{A graph $G=(V,E)$, where $E=\{e_1,\ldots,e_m\}$}
    \Parameters{A real $p\in (0, 1)$ and integers $\ell, \delta\in \N$}
    \Maintains{A $c$-certificate of $G$ given by the graph $D\cup H_\ell$ as defined below}
    \vspace{\funcSpace}
    \Proc{\FInit{}}{
        Initialize graphs $G_0, \dots, G_\ell$ all equal to $G$\;
        Initialize the empty graph $D$\;
        $H_0\longleftarrow (V,\emptyset)$\;
        $s\longleftarrow\lceil pm\rceil$\;
        \For{$i=1$ \KwTo $\ell$}{
            $r_i\longleftarrow$ a 2-independent random number generator $\{1,\ldots,s\}\to\{1,\ldots,m\}$\;
            $H_i\longleftarrow H_{i-1}\cup (V,\{e_{r_i(1)},e_{r_i(2)},\ldots,e_{r_i(s)}\})$\;
        }
        \FClean{$\emptyset$}\;
        
    }
    \vspace{\funcSpace}
    \Proc{\FDel{$e$}}{
        Delete $e$ from $G$ and $D$\;
        \FClean{$\{e\}$}\;
    }
    \vspace{\funcSpace}
    \tcc*[h]{Internal deletion of set of edges $A$, maintains $D$, $\{H_j\}_{j=0}^\ell$ and $\{G_j\}_{j=0}^\ell$}\\
    \Proc{\FClean{$A$}}{
        \For{$j=0$ to $\ell$}{
            Delete the edges of $A$ from $G_j$ and $H_j$\;
            \While{there exists an edge $g\in H_j$ contained in a cut of $H_j$ of size $<c$}{
                Delete $g$ from $H_j$\;
            }
            \While{there exists a component $C$ of $H_j$ with $S:=\partial_{G_j}(C)$ satisfying $0<\abs{S}<\delta$}{
                Add $S$ to $A$ and to $D$\;
                Delete $S$ from $G_j$\;
            }
        }
    }
  
  \caption{\label{alg:certificate_sketch}Abstract algorithm for maintaining a $c$-certificate decrementally.}
\end{algorithm}

The correctness of this approach follows by Lemmas~\ref{c-certificate}~and~\ref{l:hldcert}
applied to each subsequent version of the graph $G$.
If the certificate is not revealed to the user, and is only used to answer $c$-edge-connectivity queries or track $c$-edge-connected components, randomness is not leaked as long as the algorithm gives correct answers (which happens with high probability). By suitably increasing the constants hidden in Lemma~\ref{c-certificate} we obtain high probability correctness for \emph{all} the $O(m)$ versions of the graph.

In the following we assume that $\ell=\Theta(\log{n})$, $p=O(1/\log{n})$
and $p\delta =\Omega(c)$ so that Lemma~\ref{c-certificate} implies that $H_\ell\cup D$ remains a $c$-certificate for $G$.
\begin{lemma}\label{certificate-size}
The graph $H_\ell\cup D$ has initially $O(mp\log{n}+n\delta\log{n})$ edges and undergoes $O(n\delta\log{n})$ edge insertions throughout. 
\end{lemma}
\begin{proof}
The bound on the initial size of $H_\ell$ follows easily by the used sampling scheme. Moreover, $H_\ell$ is decremental, whereas the set $D$ can undergo both insertions (when an edge is removed from $G_\ell$) and deletions (when an edge deletion is issued to $G$). Therefore,
we only need to prove that $D$ undergoes $O(n\ell\delta)$ insertions throughout.
To this end, we show that each $G_i$ undergoes $O(n\delta)$ edge removals
following a detection of a component $C$ of $H_i$ with $0<|\partial_{G_i}(C)|<\delta$.
Recall that $H_i$ and~$G_i$ are both decremental, so at most $2n-1$ different components can ever arise in $H_i$. Each such component causes at most $\delta$ insertions to $D$ if its boundary size ever drops below $\delta$.
\end{proof}
\subsection{Supporting data structures}\label{sec:supporting}
Now we define a few data structures that we use
as subroutines when maintaining the certificate. These results are either known or should be considered folklore. For completeness, we provide the proofs of the lemmas in this section in the Appendix.

\paragraph{Restricted fully-dynamic connectivity.} Suppose $G$ is a graph
subject to edge insertions and deletions. However, assume insertion of an edge
$\{u,v\}$ is allowed only if $u$ and $v$ are currently connected.
As a result, the connected components of $G$ are decremental in time
in the sense that they can only split, but never merge.
In this restricted setting we can explicitly maintain the connected
components of each vertex and thus support constant-time connectivity
queries.

\begin{restatable}{lemma}{fddc}\label{fully-dynamic-decrmental-components}
Let $G$ be a graph subject to edge insertions and deletions. Suppose
the endpoints of each edge inserted are connected in $G$ immediately
prior to the insertion. Let $m$ be the number of initial edges in $G$
plus the number of insertions issued. There is a data structure that
maintains the connected components $\mathcal{C}=\{C_1,\ldots,C_k\}$ of $G$, and
an explicit mapping $q:V\to \{1,\ldots,|\mathcal{C}|\}$ such that $v\in C_{q(v)}$.
Moreover, after each edge deletion that increases the number of
components of $G$, the data structure outputs a pair $(j,A)$
describing how $\mathcal{C}$ evolves: the component $C_j$
is split into $C_j\setminus A$ and $A$, where $|A|\leq |C_j\setminus A|$,
and we set $C_j:=C_j\setminus A$ and $C_{k+1}:=A$, and update $k\gets k+1$.
The total update time is $O(m\log^2{n})$, whereas the sum of sizes
of sets $A$ output is $O(n\log{n})$.
\end{restatable}

\paragraph{Maintaining boundaries of splitting sets in a fully dynamic graph.}
We will often need to solve the following abstract dynamic problem on graphs.
Suppose we have two possibly unrelated graphs: a \emph{fully dynamic} graph $G$ and a \emph{decremental} graph $H$, both on $V$.
We would like to maintain boundaries $\partial_G(C)$ of all the connected components
$C$ of $H$ under the allowed updates to $G$ and $H$.

\begin{restatable}{lemma}{fddb}\label{fully-dynamic-decremental-boundary}
    Let $G=(V,E)$ be a fully dynamic graph.
    Let $\mathcal{C}$ be the set of connected components of some (possibly unrelated) decremental graph on $V$.
    Suppose the updates to $\mathcal{C}$ are given in the same form
    as in the output of the data structure of Lemma~\ref{fully-dynamic-decrmental-components}.
    Then, the boundaries $\partial_G(C)$ for $C\in\mathcal{C}$ can be maintained explicitly subject to edge insertions/deletions issued to $G$, and updates to $\mathcal{C}$
    in $O((n+m)\log{n})$ total time, where $m$ is the number
    of initial edges of $G$ plus the number of edge insertions issued to $G$.
\end{restatable}

\subsection{Basic data structure}
We now discuss how the algorithm maintaining the $c$-certificate can be efficiently implemented.
We start with a basic version of the data structure that does not yet achieve linear dependence on $m$.

First consider maintaining the graphs $H_i$. Recall that we need to efficiently detect cuts of size $<c$ in $H_i$ under deletions, prune $H_i$ of these cuts, and keep track of how the $c$-edge-connected components (or, equivalently, the connected components) of $H_i$ evolve.
To this end, we will need the following auxiliary dynamic graph data structures.
First of all, we maintain a $c$-certificate
of $H_i$ using the data structure of Theorem~\ref{t:thorup}.\footnote{We could in principle use a \emph{decremental} $c$-edge-cut data structure on the graph $H_i$ itself as opposed to a \emph{fully dynamic} data structure on its $c$-certificate, but that would prove less efficient.}
On top of the $c$-certificate of $H_i$, we set up a fully-dynamic $c$-edge-cut data structure,
and the data structure of Lemma~\ref{fully-dynamic-decrmental-components}.
Since the $c$-edge-connected components of $H_i$
are precisely the components of the graph obtained by repeatedly removing $<c$-edge cuts from the $c$-certificate of $H_i$,
these components combined can maintain
the $c$-edge-connected components of $H_i$ and provide
an efficient description of the splits these
components undergo.

In the basic version of our algorithm, for each $G_i$ in turn, we use a separate decremental boundary maintenance data structure of Lemma~\ref{fully-dynamic-decremental-boundary},
where the connected components whose boundaries we care about come from~$H_i$.
This data structure is passed all the updates to the components of $H_i$
as described in Lemmas~\ref{fully-dynamic-decrmental-components}~and~\ref{fully-dynamic-decremental-boundary}.
Recall how the boundary maintenance structures are used: while for some ($c$-edge-) connected component $C$ of $H_i$, the boundary of $C$ in $G_i$ has positive size $<\delta$, we remove that boundary from~$G_i$ (and propagate that change to subsequent layers $j>i$). In particular,
for $i=0$, since $H_0$ is empty and has only trivial components,
this corresponds to removing from $G_0$ all edges incident to vertices of degree $<\delta$ until no such vertices exist. The boundaries of size less than $\delta$
can be accessed easily using the data structure of~Lemma~\ref{fully-dynamic-decremental-boundary} associated with $G_i$.

Assuming a fully-dynamic $c$-edge-cut data structure with amortized update time $T_c(n)$ and a proper choice of parameters $p,\delta$,
the above algorithm
runs in $\Ot(m)$ time and, most importantly, updates
the maintained $c$-certificate a sublinear (in $m$) number of times.

\begin{lemma}\label{reduction-simple}
    There exists a decremental algorithm maintaining a $c$-certificate for $G$
    such that the certificate undergoes $O(nc\log^4{n})$ edge updates throughout. 
    The total update time of the algorithm is $\Ot(m)+O(n(c+\log{n})\cdot T_c(n)\log^3{n})$ with high probability.
\end{lemma}
\begin{proof}
    We set $p=\frac{1}{\log^3{n}}$. This forces us to set $\delta=12c\log^3{n}$
    for $p\delta$ to be sufficiently large which is required by Lemma~\ref{c-certificate}.
    Recall that each $H_i$ has $O(mp\ell)$ edges initially and we maintain
    its $c$-certificate under edge deletions
    using Theorem~\ref{t:thorup}.
    As a result, this incurs a cost of $O(mp\ell\log{n}+n\cdot(c+\log{n})\cdot T_c(n)\log^2{n})$ time.
    Since the $c$-certificate of $H_i$ undergoes $O(n(c+\log{n}))$ updates,
    using a fully-dynamic $c$-edge-cut data structure upon the maintained $c$-certificate of $H_i$ costs ${O(n(c+\log{n})\cdot T_c(n))}$
    time.
    Similarly, using the data structure of Lemma~\ref{fully-dynamic-decrmental-components} on the $c$-certificate of $H_i$
    costs $O(n(c\log{n})\log^2{n})$ time.
    Summing over all $\ell=O(\log{n})$ graphs $H_i$, we get
    $O(mp\log^3{n}+n(c+\log{n})\cdot T_c(n)\log^3{n})$ time.
    By our choice of $p$, the first term
    is $O(m)$.
    
    We also use a simple-minded boundary maintenance data structure of Lemma~\ref{fully-dynamic-decremental-boundary} on each of $O(\log{n})$ graphs
    $G_i$ with components from~$H_i$. The total update time of these
    data structures is $O(m\log^2{n})$.
    The bound on the number of updates to the certificate
    follows by Lemma~\ref{certificate-size}.
\end{proof}

Note that the only obstacle preventing us from getting an $O(m)$ bound in Lemma~\ref{reduction-simple} is the maintenance 
of component boundaries for each pair $(H_i,G_i)$ independently.
The simple-minded solution yields an $O(m\log^2{n})$ overhead for this task
and in fact solves an overly general problem of maintaining
the boundaries regardless of their size: recall that we only care
about the precise elements of the set $\partial_{G_i}(C)$ for a component $C$ of $H_i$ if $|\partial_{G_i}(C)|\leq \delta$. 
Otherwise, a (high-probability) guarantee that $|\partial_{G_i}(C)|>\delta$ is sufficient
for our needs.

\subsection{Maintaining small boundaries}
We now describe how to obtain an $O(m)+\Ot(n\delta)$ bound for maintaining all the required
small boundaries. Note that this will imply the desired $O(m)$ running time
for sufficiently dense graphs, assuming $T_c(n)$ is low enough. First of all, let $\Delta_i:=G_i\setminus G_\ell$. 
We can split the task of maintaining $\partial_{G_i}(C)$ into maintaining
$\partial_{G_\ell}(C)$ and $\partial_{\Delta_i}(C)$ separately.
Clearly, we have $\Delta_i\subset D$.
Recall that $D$ (and thus also $\Delta_i$) is initially empty and undergoes only $O(n\delta\log{n})$ insertions throughout.
As a result, we can afford maintaining the boundaries of the form $\partial_{\Delta_i}(C)$ even exactly (i.e., regardless of their sizes) using the 
data structure of Lemma~\ref{fully-dynamic-decremental-boundary} in $O(\ell \cdot n\delta\log{n}\cdot \log{n})=O(n\delta\log^3{n})$ time.

It remains to show how to efficiently maintaining the boundaries $\partial_{G_\ell}(C)$ for components $C$ of the graphs $H_1,\ldots,H_\ell$,
provided that $|\partial_{G_\ell}(C)|\leq\delta$.
We accomplish this goal using three components.
\paragraph{The sampled graph $R$.} The first component
is responsible solely for estimating the sizes $|\partial_{G_\ell}(C)|$.
Let $R$ be a graph obtained from $G_\ell$ via uniform sampling with
probability $q=1/\log^2{n}$. Clearly, the graph $R$ has size $\Theta(mq)$
with high probability by the Chernoff bound.
Recall that $G_\ell$ is decremental; whenever an edge $e$ of $G_\ell$
is deleted, it is removed from $R$ as well if it was sampled.
So $R$ can be initialized and maintained in $O(m)$ total time.
Assume $\delta=\Omega(\log^2{n}\cdot \log{m})$, where the constant hidden
is sufficiently large.
Then, for each version of $G_\ell$ in time,
and each of some $O(\poly\{n,m\})$ sets $C\subseteq V$ chosen independently
of~$R$, $|\partial_{G_\ell}(C)|\leq \delta$ implies $|\partial_{R}(C)|\leq 2q\delta$, and $|\partial_{R}(C)|\leq 2q\delta$ implies $|\partial_{G_\ell}(C)|\leq 4\delta$,
both with high probability via the Chernoff bound.

\paragraph{The small-boundary oracle.}
Consider the following abstract problem. Suppose $G=(V,E)$ is a fully-dynamic
graph. We would like to have a data structure that supports the following query:
given some $S\subseteq V$, compute $\partial_G(S)$. The obvious query procedure
would be to go through all edges incident to the vertices of $S$; this would
give a $O(|E(S,V)|)$ query bound. However, if $|S|\cdot |\partial_G(S)|$ is significantly smaller than $|E(S,V)|$, a more efficient solution
is possible. Formally, we prove the following theorem which we believe might be of independent interest.

\begin{theorem}\label{small-boundary-oracle}
    Let $G=(V,E)$ be an initially empty graph subject to edge insertions and deletions and let $s$, $1\leq s\leq n$, be an integral parameter. There exists a data structure that can process up to $O(\poly{n})$ queries about the current set $\partial_G(S)$, where $S\subseteq V$ is the query parameter, so that with high probability, each query is answered correctly in $O\left(|S|s+|E(S,V)|\cdot \frac{|\partial_{G}(S)|}{s}+\log{n}\right)$ time.
    The data structure is initialized in $O(ns)$ time and can be updated
    in constant time.
\end{theorem}

We use the above data structure for $G_\ell$ and only ask queries when
$|\partial_{G_\ell}(S)|=O(\delta)$. By setting $s=\delta\cdot \log^2{n}$
we will achieve $O(|S|\delta\log^2{n}+|E(S,V)|/\log^2{n})$ query time.
Since $G_\ell$ undergoes $m$ updates, the total update time of this
data structure is $O(m+n\delta\log^2{n})$.

\paragraph{Maintaining small boundaries of $G_\ell$ under splits.}
Finally, we show how to combine the above two components with
a neat variant of the data structure of Lemma~\ref{fully-dynamic-decremental-boundary}
in order to maintain, for each $i$, the boundaries
of components $C$ of $H_i$ with $|\partial_{G_\ell}(C)|\leq \delta$
in $\Ot(n\delta)+O(m/\log{n})$ time with high probability.
Through all $i$, this will imply the desired $\Ot(n\delta)+O(m)$ total
time bound.

Let $\mathcal{C}=\{C_1,\ldots,C_k\}$ be the components of $H_i$.
Recall that the data structure of Lemma~\ref{fully-dynamic-decrmental-components}
for the $c$-certificate of $H_i$ yields decremental updates of the form $(j,C')$ (where $\emptyset \neq C'\subseteq C_j$ and $|C'|\leq |C_j\setminus C'|$) to $\mathcal{C}$ that set $C_j:=C_j\setminus C'$
and $C_{k+1}:=C'$. As argued in Lemma~\ref{fully-dynamic-decrmental-components}, the total number
of updates is at most $n-1$ and the sum of $|C'|$ over all updates is $O(n\log{n})$.
We will show how to process these updates
so that with high probability, at all times for each $C_j$, $|\partial_{G_\ell}(C_j)|\leq \delta$ implies
that we store the set $\partial_{G_\ell}(C_j)$ explicitly.
Wlog. assume that $\mathcal{C}=\{V\}$ initially.

We will say that a set $S\subseteq V$ \emph{has small boundary} if
$|\partial_R(S)|\leq 2q\delta$. Otherwise, we say that $S$ \emph{has large boundary}. As argued before, with high probability,
the subset of components with small boundary includes all components
$C$ or our interest, i.e., with $|\partial_{G_\ell}(C)|\leq \delta$, and
does not include any components with $|\partial_{G_\ell}(C)|=\omega(\delta)$.
We keep track of which components have large/small boundaries by running a simple-minded
boundary maintenance data structure of Lemma~\ref{fully-dynamic-decremental-boundary} on $R$.
The total update time of this data structure is $O(n\log{n}+|R|\log{n})=O(n\log{n}+m/\log{n})$ whp.

A naive approach to solve our problem would be to maintain
$\partial_{G_\ell}(C)$ for small boundary components $C\in\mathcal{C}$.
However, this approach fails for the following reason. Suppose
a component $C$ is split into $C',C''$, where $|C'|\leq |C''|$.
Assume that $C''$ does have a small boundary, whereas $C$ and $C'$ do not.
It is then unclear how to compute the set $\partial_{G_\ell}(C'')$ (of size $\leq \delta$) using time
less than linear in either $|C''|$ or $|E(C',V)|$. Had $\Omega(n)$ splits
like this happened, we could spend time as much as either $\Theta(n^2)$ or
$\Theta(m\log{n})$ which is obviously too much.
We need a smarter approach.

First of all, denote by $\mathcal{S}$ be the family of sets $C$ that \emph{ever}
appeared in $\mathcal{C}$ and had small boundary when still in $\mathcal{C}$.
We will maintain $\partial_{G_\ell}(S)$ for \emph{all} $S\in\mathcal{S}$ -- as opposed
to exclusively for $S\in \mathcal{S}\cap\mathcal{C}$ as in the naive approach.
Moreover, for each 
$C\in\mathcal{C}$, $C\neq V$, we store (a pointer to) $s(C)$:
the unique smallest set in $\mathcal{S}$ such that $C\subsetneq s(C)$.
Initially we have 
$\mathcal{C}=\mathcal{S}=\{V\}$, and $\partial_{G_\ell}(V)=\emptyset$.

We now show how to update
the stored information when an update $(j,C')$ comes.
Let $A=C'$ and $B=C_j\setminus C'$. Recall that $|A|\leq |B|$.
Then, if $C_j\in \mathcal{S}$, we have $s(A)=s(B)=C_j$.
Otherwise, we have $s(A)=s(B)=s(C_j)$.
Now, if some $C\in\mathcal{C}$ becomes small-boundary\footnote{Recall that the boundaries $\partial_R(C)$ are maintained explicitly using the simple-minded data structure
of Lemma~\ref{fully-dynamic-decremental-boundary}. Consequently, it is easy to detect this event ``on the fly''.} 
(either as a result
of edge deletion issued to $R$ or immediately when it appears),
we compute $\partial_{G_\ell}(C)$ as follows.
If $|C|\leq |s(C)|/2$, then we compute $\partial_{G_\ell}(C)$
using the small boundary oracle query on $C$.
Otherwise, we compute it by issuing a query about the set $s(C)\setminus C$ to
the small boundary oracle and then taking the
symmetric difference $\partial_{G_\ell}(s(C))\triangle \partial_{G_\ell}(s(C)\setminus C)$ which equals $\partial_{G_\ell}(C)$.
It is important to note that we do not require that $s(C)\setminus C$
is an element of $\mathcal{S}$ here; it is sufficient to have that $|\partial_{G_{\ell}}(s(C)\setminus C)|=O(\delta)$,
which follows by $\partial_{G_\ell}(s(C)\setminus C)\subseteq \partial_{G_\ell}(C)\cup \partial_{G_\ell}(s(C))$ and $|\partial_{G_\ell}(C)|,|\partial_{G_\ell}(s(C))|\leq \delta$ with high probability.
Clearly, taking the symmetric difference takes $O(\delta)$ time.

\begin{lemma}\label{l:small-bnd}
    The total time spent on computing all the required boundaries $\partial_{G_\ell}(C)$ for $C\in\mathcal{S}$
    is $O(n\delta\log^3{n}+m/\log{n})$ with high probability.
\end{lemma}
\begin{proof}
    Consider a natural tree that the sets of $\mathcal{S}$ form, where each
    $C\neq V$ is a child of $s(C)$. Let $\mathcal{S}^*=\{C\in\mathcal{S}:
    |C|\leq |s(C)|/2\}$.
    Computing boundaries of sets
    $C\in \mathcal{S}$ such that $|C|\leq |s(C)|/2$, i.e., of sets $C\in\mathcal{S}^*$,
    requires a single oracle query for $\partial_{G_\ell}(C)$ per each $C\in\mathcal{S}^*$.
    By Theorem~\ref{small-boundary-oracle}, such a query costs
    \begin{equation*}
        O\left(|C|\delta\log^2{n}+\sum_{v\in C}\frac{\deg(v)}{\log^2{n}}+\log{n}\right)
    \end{equation*}
    time.
    When $|C|>|s(C)|/2$, $s(C)\setminus C$ equals the union of siblings of $C$ in the tree that the sets of $\mathcal{S}$ form.
    So the cost of a query for $\partial_{G_\ell}(s(C)\setminus C)$ is:
    \begin{equation*}
        O\left(\sum_{C'\text{ a sibling of C}}\left(|C'|\delta\log^2{n}+\sum_{v\in C'}\frac{\deg(v)}{\log^2{n}}+\log{n}\right)\right)
    \end{equation*}
    Observe that each $C'\in \mathcal{S}$ has at most one sibling whose size is at least $|s(C')|/2$.
    As a result, each $C'$ with $|C'|\leq |s(C')|/2$ contributes to a sum above for at most one $C\in \mathcal{S}$ with $|C|>|s(C)|/2$.
    As a result, the total time spent on small boundary oracle
    queries (through all $C\in\mathcal{S}$) is:
    $$O\left(\sum_{C\in\mathcal{S}^*} \left(|C|\delta\log^2{n}+\sum_{v\in C}\frac{\deg(v)}{\log^2{n}}\right)\right)=O\left(\sum_{v\in V}\left( \delta\log^2{n}+\frac{\deg(v)}{\log^2{n}}\right)\cdot |\{C\in\mathcal{S^*}:v\in C\}|\right).$$
    Observe that each $v\in V$ can be an element of at most $O(\log{n})$
    sets of $\mathcal{S}^*$: all these sets are ancestors of $\{v\}$
    in the tree corresponding to $\mathcal{S}$ and have size smaller
    than their respective parent by a factor of at least $2$.
    As a result, the total time spent on this step
    is
    $O(n\delta\log^3{n}+\sum_{v\in V}\deg(v)/\log{n})=O(n\delta\log^3{n}+m/\log{n})$.
    This dominates the $O(n\delta\log{n})$ cost of taking symmetric differences.
\end{proof}
We also need to maintain the stored boundaries
$\partial_{G_\ell}(C)$ for $C\in\mathcal{S}$ under edge deletions
that $G_\ell$ undergoes. However, to avoid spending $O(m)$ time
per single $H_i$ on this, for this part we need to consider all the graphs $H_i$
simultaneously. Note that for a fixed $i$, $\mathcal{S}$ only
grows and contains at most $2n-1$ elements.
Hence, the total size of the stored sets $\partial_{G_\ell}(C)$,
$C\in\mathcal{S}$, is $O(n\delta)$ (whp).
For each $e\in E(G_\ell)$ we maintain a list of pointers
to such stored boundaries with $e\in \partial_{G_\ell}(S)$, through all $H_i$.
The total number of pointers ever inserted into these lists is clearly $O(n\delta\ell)=O(n\delta\log{n})$.
When an edge $e$ is removed from $G_\ell$, we scan the attached list
of $e$ and remove this edge from the required boundaries it was contained in.
The total time spent on this can be seen to be no more than the
total number of insertions into the lists, i.e., $O(n\delta\log{n})$.

\thmcertificate*
\begin{proof}
Recall that the simple reduction from Lemma~\ref{reduction-simple} had
$O(m)+O(n(c+\log{n})\cdot T_c(n)\log^3{n})$ operation cost of the fully-dynamic
$c$-edge connected components data structures. It also required
setting $\delta$ to at least $c\log^3{n}$.
The cost of maintaining the needed small boundaries is dominated by the application of Lemma~\ref{l:small-bnd} for each $i=1,\ldots,\ell$.
The statement about the running time of the algorithm follows.

We now turn to proving the statement that the algorithms offers a final self-check after processing all updates.
Notice that we only need to check that no edge between distinct $c$-edge-connected components of the certificate was missing at any point throughout the execution of the algorithm; indeed, even if an edge with both endpoints in the same $c$-edge-connected component of the certificate was missing that wouldn't affect any answers to $c$-edge-connectivity queries on the certificate.

First note that in a correct execution of our algorithm all edges between distinct $c$-edge-connected components of $H_l$ (and hence, of the certificate) are always present in the certificate. 
That is, we only need to verify that no edge between two distinct ($c$-edge-) connected components of $H_l$ is added and that each deleted edge from $G$ is either present in the certificate or both of its endpoints belong to the same ($c$-edge-) connected component of $H_l$.
We assume that each vertex has access to the ID of its ($c$-edge-) connected component in $H_l$, so that we can check in constant time whether the two vertices are part of the same ($c$-edge-) connected component. 
These IDs are provided by invocation of the Lemma \ref{fully-dynamic-decrmental-components} on the certificate. 
Whenever an edge $e$ is deleted from the graph $G$ (and hence from the certificate), we simply check that $e$ is part of the certificate if its endpoints are in distinct ($c$-edge-) connected components of $H_l$; if that is not the case, we mark the execution of the algorithm invalid, as edge $e$ should have been part of the certificate. On the other hand, if both the endpoints of a deleted edge $e$ were part of the same ($c$-edge-) connected component of $H_l$, then no query might have been answered incorrectly. 

Finally, edges might be added to the certificate due to the update in our data structures following an edge deletion from $G$. Again, we need to make sure that no edge is added to the certificate that was supposed to be there before the edge deletion and is omitted due to an error.
Specifically, for the edges added to the certificate we need to check that both of their endpoints are in the same ($c$-edge-) connected component of $H_l$ right before the last edge deletion from $G$ (which potentially caused the splitting of connected components of $H_l$). If that is not the case, then we again flag the execution invalid as the endpoints of these edges were part of distinct ($c$-edge-) connected components of $H_l$ and should already by part of the certificate.
Notice that the splits of ($c$-edge-) connected components of $H_l$ are described by the output of the data structure of Lemma \ref{fully-dynamic-decrmental-components}, and hence the queries can be answered efficiently (even an $O(\log n)$ bound per query would be enough to keep the running time withing the stated bound due to the limited number of updates to the certificate).
\end{proof}

\subsection{Small boundary oracle}\label{subsec:finding_boundary}
In this section we prove Theorem~\ref{small-boundary-oracle}.
Recall that the goal is to have a data structure that maintains
a fully dynamic graph $G=(V,E)$ and supports queries regarding $\partial_{G}(S)$,
where $S\subseteq V$ is a query parameter.

First of all, we will leverage the well-known XOR trick~\cite{DBLP:conf/soda/AhnGM12,DBLP:conf/pods/AhnGM12} for deciding
if a boundary of some subset of vertices is non-empty.
We now briefly describe this method.
Suppose each $e\in E$ is assigned a random bit-string $x_e$ of length $\Theta(\log{n})$ that fits in $O(1)$ machine words.
Let $x_v=\bigoplus_{vw=e\in E} x_e$ denote the XOR of the respective bit-strings of edges incident to $v$.
Then, one can prove that, given $S\subseteq V$, with high probability the XOR $\bigoplus_{u\in S}x_u$ is non-zero if and only if $\partial_{G}(S)\neq \emptyset$. So, emptiness of $\partial_G(S)$ can be tested in $O(|S|)$ time.

Let $s\geq 1$ be an integral parameter. The main idea is as
follows. We partition the edge set $E$ into $E_1,\ldots,E_s$.
Each $e\in E$ is assigned to one of these sets uniformly at random.
Let us apply the XOR-trick for each $E_i$ separately.
To this end, now $x_v$ is a vector of $s$ bit-strings,
where $x_v(i)=\bigoplus_{vw=e\in E_i} x_e$.
Given that,
in $O(s|S|)$ time we can find the set $I$ of all $i$ such that $\partial_G(S)\cap E_i\neq\emptyset$ (whp).
Clearly, in order to find $\partial_G(S)$, we only need to look
for this boundary's elements in $\left(\bigcup_{i\in I}E_i\right)\cap E_G(S,V)$.
If $|\partial_G(S)|$ is small compared to $s$, one can prove that, with high probability,
this strategy is more efficient than iterating through the entire set $E(S,V)$.
We prove this formally below.

\begin{lemma}\label{lem:find_boundary_works}
	Let $S\subseteq V$. Then, with high probability,
	the query procedure computes $\partial_G(S)$ correctly in
	$O\left(s|S|+|E_G(S,V)|\cdot \frac{|\partial_G(S)|}{s}+\log{n}\right)$ time.
\end{lemma}

\begin{proof}
	
	Let the bit-strings $x_e$ consist of $\gamma=O(1)$ machine words, each with
	at least $\lceil\log_2{n}\rceil$ bits. As argued before,
	computing the bit-strings $y(i)=\bigoplus_{v\in S}x_v(i)$ for all
	$i=1,\ldots,s$ and finding the set $I=\{i:y(i)\neq 0\}$ takes
	$O(s|S|)$ worst-case time.
	
	Now, let us consider the number of edges searched. Suppose that $e\in E_{G}(S, S)$ has endpoints $u, v\in S$ and belongs to $E_j$. Then $x_e$ does not contribute to $y(j)$ as it is present twice in the XOR, once in $x_v(j)$ and once in $x_u(j)$. Hence, if $y(j)\neq 0$, i.e., $j\in I$, then there is an edge of $E_j$ that has only one endpoint in $S$ and thus belongs to $\partial_{G}(S)$. Since each $e\in \partial_{G}(S)$ contributes to a single element of $y$, $\abs I\leq \abs{\partial_{G}(S)}$. 
	Now, for $e\in E$, let $Y_e$ be the indicator of the event $(e\in \bigcup_{i\in I}E_i)$, i.e., $Y_e=1$ if $e\in \bigcup_{i\in I}E_i$ and $Y_e=0$ otherwise. The set $I$ is entirely determined by the variables $(x_e)_{e\in \partial_{G}(S)}$, so for the remaining edges, $E_{G}(S, V)\setminus \partial_{G}(S)=E_{G}(S, S)$, the random variables $\{Y_e\}_{e\in E_{G}(S, S)}$ are mutually independent and independent of the choice of $I$ except for its size, $\abs I$. It follows that the sum $Y=\sum_{e\in E_{G}(S, S)}Y_e$ satisfies 
	$$\E Y = \abs{E_G(S, S)}\cdot \frac{|I|}{s}\leq \abs{E_{G}(S, V)}\cdot \frac{|\partial_{G}(S)|}{s}.$$ 
	Since $Y$ is a sum of independent random variables, we may write $\mu = \abs{E(S, V)}\cdot \abs {\partial_{G}(S)}/s$.
	Let $t>0$ be a constant and $\mu'=\mu+3t\log{n}$, and apply the Chernoff bound of \cref{thm:chernoff} to get
	\begin{align*}
		\Pr\left[ Y> 2\mu + 6t\log(n) \right]=\Pr[Y>(1+1)\mu']\leq \exp(-\max\{\mu, 3t\log(n)\}/3)\leq n^{-t}.
	\end{align*}
	Thus,
	with high probability the number of edges checked by the algorithm is
	$$O\left(|\partial_G(S)|+|E_G(S,V)|\cdot \frac{|\partial_G(S)|}{s}+\log{n}\right),$$
    and, as a result, the running time of the query procedure is
	$$O\left(s|S|+|\partial_G(S)|+|E_G(S,V)|\cdot \frac{|\partial_G(S)|}{s}+\log{n}\right).$$
    To obtain the desired bound note that if $s\geq |\partial_G(S)|$,
    then the term $|\partial_G(S)|$ above is dominated by $s|S|$,
    and otherwise it is dominated by the third term.

	Finally, let us consider the probability that the output of the query is correct. It is not hard to see that the output is correct if and only if the set $I$ corresponds to the set $J=\{j\in [s]\mid E_j\cap \partial_{G}(S)\neq \emptyset\}$, since in that case, the algorithm searches all groups containing an edge of $\partial_{G}(S)$. We have already established that $I\subset J$. So let $j\in J$ be given. Then there is some edge $e\in \partial_{G}(S)\cap E_j$. We have:
	\begin{align*}
		y(j) = \bigoplus_{v\in S}x_v(j) = \bigoplus_{f\in \partial_{G}(S)\cap E_j} x_{f}
	\end{align*} 
	since for every edge $e$ of $E_{G}(S, S)$, $x_e$ appears twice in the XOR. The probability that $y(j)=0$, or equivalently $j\not\in I$, is hence the probability that $x_e = \bigoplus_{f\in (\partial_{G}(S)\setminus\{e\})\cap E_j} x_f$. Since $x_e$ is independent of the right-hand side, this probability is exactly $2^{|x_e|}\leq 2^{-\gamma \log_2{n}}\leq n^{-\gamma}$. By a union bound, it follows that $J\subset I$ with probability at least $1-\abs J\cdot n^{-\gamma}\leq 1-sn^{\gamma}\leq 1-n^{\gamma-1}$.
\end{proof}

When an edge $e=uv$ is inserted into $G$, all we have to do is
pick a random set $E_j$ for $e$, sample a random bit-string $x_e$,
and update $x_w(j):=x_w(j)\oplus x_e$ for $w\in \{u,v\}$.
To handle the deletion of $e$, all we have to do
is to repeat the last step of insertion and remove $e$ from $E_j$.
So an edge update can be clearly performed in $O(\gamma)=O(1)$ worst-case time.
The data structure can be initialized in $O(ns+m)$ time
by first filling the values $x_v(i)$ with zeros and then inserting
all the edges.

Finally, to guarantee high-probability correctness and query time bounds
for $\poly{(n)}$ queries, it is enough to set constants $\gamma$ and $t$ (from the proof of Lemma~\ref{lem:find_boundary_works}) sufficiently large.
The full pseudocode of the data structure is given in Algorithm~\ref{alg:find_boundary}.

\begin{algorithm}[ht]
    \SetKwInOut{Input}{Input}
    \SetKwInOut{Parameters}{Parameters}
    \SetKwInOut{Maintains}{Maintains}
    
    \SetKwProg{Fn}{Function}{:}{}
    \SetKwProg{Proc}{Procedure}{:}{}
    \SetKwProg{PFn}{Private Function}{:}{}
    \SetKwFunction{FInit}{Initialize}
    \SetKwFunction{FDel}{Delete}
    \SetKwFunction{FIns}{Insert}
    \SetKwFunction{FFind}{FindBoundary}
    
    \Input{A graph $G=(V, E)$ on $n$ vertices}
    \Parameters{Positive integers $\gamma=O(1)$ and $s\leq n$.}
    \vspace{\funcSpace}    
    
    \Proc(){\FInit{}}{
        Initialize $s$ sets $E_1, \dots, E_s$\;
        Fill values $x_v(j)$ for $v\in V$ and $j\in [s]$ with all-zero bit-strings of length $\gamma\cdot \lceil\log_2{n}\rceil$\;
        \For{$e\in E$}{
           \FIns{$e$}\;
        }
    }
    \vspace{\funcSpace}
    \Proc(){\FIns{$e=\{u, v\}$}}{
        Insert $e$ into some $E_j$ of $E_1, \dots, E_s$ uniformly at random\;
        Let $x_e\in \{0, 1\}^{\gamma \cdot \lceil\log_2{n}\rceil}$ be a bit-string chosen uniformly at random\;
        Let $x_v(j):=x_v(j)\oplus x_e$\;
        Let $x_u(j):=x_u(j)\oplus x_e$\;
    }
    \vspace{\funcSpace}
    \Proc(){\FDel{$e=\{u, v\}$}}{
        Let $E_j$ be the set containing $e$\;
        Delete $e$ from $E_j$\;
        Let $x_v(j):=x_v(j)\oplus x_e$\;
        Let $x_u(j):=x_u(j)\oplus x_e$\;
    }
    \vspace{\funcSpace}
    \Fn(\tcc*[h]{Find $\partial_G(S)$ for a subset $S\subset V$})
    {\FFind{$S$}}{
        Let $B := \emptyset$\;
        Let $y:= \bigoplus_{v\in S} x_v$\;
        Let $I := \{i\in [s]\mid y(i)\neq 0\}$;

        \For{$uv=e\in \bigcup_{i\in I}(E_i\cap E(S, V))$ with $u\in S$}{
            \lIf{$v\notin S$}{
                $B:=B\cup \{e\}$
            }
        }
        \Return{B}
    }
      
  \caption{\label{alg:find_boundary}Small boundary oracle.}
\end{algorithm}

\section{Decremental $c$-Edge-Connectivity}
In this section we briefly explain how Theorem~\ref{thm-certificate} implies
decremental $c$-edge-connectivity algorithms with $O(m)$ total update time
for sufficiently dense graphs.

It is important to note at this point that there are two settings that might be of interest. First, we might want to have a decremental algorithm maintaining
$c$-edge-connected \emph{components}, that is, supporting queries whether
two vertices belong to the same $c$-edge-connected component of $G$.
However, we might alternatively want to have a decremental algorithm
maintaining the $c$-edge-connected \emph{classes}, i.e., supporting
queries whether there exist $c$ edge-disjoint paths between
some two vertices.
Recall that these settings are equivalent for $c=1,2$, but differ for $c\geq3$:
then a pair of $c$-edge-connected vertices might not belong to
the same $c$-edge-connected component.

Let us first consider the decremental $c$-edge-connected components problem. Then we have:
\begin{restatable}{theorem}{thmgeneralresuctioncomponent}\label{thm:general-reduction-components}
There exists a Monte Carlo randomized decremental $c$-edge-connected components algorithm 
with
$O(m+nc(\log^7{n}+\log^4{n}\cdot T_c(n)))$ total update time. The algorithm is correct with
high probability.
\end{restatable}
\begin{proof}
We maintain a $c$-certificate $H$ of $G$ using Theorem~\ref{thm-certificate}.
Observe that $H$ has the same $c$-edge-connected components
as $G$. 
We additionally maintain a fully-dynamic $c$-edge-cut data structure
for $H$, and a data structure of Lemma~\ref{fully-dynamic-decrmental-components}.
These two combined allow us to prune the certificate from $<c$-edge-cuts
and explicitly maintain the $c$-edge-connected components of $H$ (which enables constant-time queries about the component a vertex belongs to).
Since~$H$ undergoes only $O(nc\log^4{n})$ edge updates, and each can be processed in $O(T_c(n)+\log^2{n})$
amortized time, the theorem follows.
\end{proof}
By plugging in the specific known upper bounds on $T_c(n)$ for $c=1,2$, we obtain:

\thmfullydynamicconn*

For $c\geq 3$,  $T_c(n)=O(n^{1/2}\poly{(c)})$ has been proved~\cite{Thorup07}, and therefore for $c=O(n^{o(1)})$ the $c$-edge-connected components can be maintained under deletions in $O(m)+\Ot(n^{3/2+o(1)})$ total time.
\thmccomponents*

Now consider the decremental $c$-edge-connected classes problem., i.e., decremental pairwise $c$-edge-connectivity.
\begin{restatable}{theorem}{thmgeneralresuctionclasses}
\label{thm:general-reduction-classes}
Suppose there exists a fully-dynamic $c$-edge-cut algorithm with $T_c(n)$ amortized update time, and a fully-dynamic pairwise $c$-edge-connectivity algorithm with $U_c(n)$ amortized update time and $Q_c(n)$ query time.
There exists a Monte Carlo randomized decremental $c$-edge-connected components algorithm 
with
$O(m+nc\log^7{n}+n(c+\log{n})T_c(n)\log^2{n}+nc\log^4{n}\cdot U_c(n))$ total update time and $O(Q_c(n))$ query time.
The algorithm is correct with high probability.
\end{restatable}
\begin{proof}
We maintain a $c$-certificate $H$ of $G$ using Theorem~\ref{thm-certificate}.
Recall that $H$ has the same $c$-edge-connected components
as $G$. 
So, we additionally maintain the certificate using the assumed fully-dynamic
$c$-edge-connected classes data structure and use it to answer queries.
\end{proof}

Jin and Sun~\cite{jin2020} have recently showed that for $c=(\log{n})^{o(1)}$, a deterministic fully-dynamic $c$-edge-connected classes data structure with $U_c(n)=O(n^{o(1)})$ and
$Q_c(n)=O(n^{o(1)})$ exists. By combining their result
with the fully-dynamic $c$-edge-cut algorithm of Thorup~\cite{Thorup07} with
\linebreak ${T_c(n)=O(n^{1/2}\poly{(c)})}$, 
we obtain the following.

\tdecrclasses*

\section{Reducing the Number of Random Bits}\label{sec:bits}
In this section we show that our algorithms can be tuned to require only $O(c\poly \log n)$ random bits over all updates.
We take advantage of the pseudorandom number generator by Christiani and Pagh~\cite{DBLP:conf/focs/ChristianiP14}, which, given only a $O(c\poly \log n)$ truly random bits, can generate $O(n)$ random numbers, such that each number is between $1$ and $n'$, $n\leq n'\leq 2n$ and the generated numbers are $\Theta(c \poly \log n)$-independent.
Generating each number takes $O(1)$ time,
whereas initialization takes $O(c\polylog{n})$ time.

The key property that we use is the fact that $\Theta(\log n)$-independence is sufficient for a Chernoff-like bounds to hold~\cite{DBLP:conf/focs/ChristianiP14}.

Our algorithm uses randomness for three purposes:
\begin{enumerate}
    \item In order to initially sample the graphs $H_i^0$ (for all $i$) and $R$.
    \item Within the data structure of Theorem~\ref{t:thorup} to maintain a $c$-certificate of each $H_i$.
    \item Within the small-boundary oracle, to partition the edges of $E$ into sets $E_1, \ldots, E_s$ .
    \item Within the small-boundary oracle, to generate the random bits associated with each edge.
\end{enumerate}

We now discuss how to implement each item using the pseudorandom generator of~\cite{DBLP:conf/focs/ChristianiP14}.

We already argued in Section~\ref{sec:certificate} that for sampling $H_1^0,\ldots,H_\ell^0$, a polylogarithmic number of random bits is sufficient. For the sampled graph $R$ we only
used Chernoff bounds for a polynomial number of sums
of indicator variables, so indeed polylogarithmic independence is enough, and the sampling can be performed using the pseudorandom generator of~\cite{DBLP:conf/focs/ChristianiP14}.

When making use of the partition in item~(3), for efficiency we only apply the Chernoff bound
to a polynomial number of sums of independent indicator variables $Y_e$ with the same mean for edges $e$ in some
subset of $E$.
Hence, $O(\polylog{n})$-independence between the variables $Y_e$ is sufficient.
The partition of $E$ can be performed by sampling
each $E_i$ to be a $\left\lceil (m-\sum_{j<i}|E_j|)/(s-i+1)\right\rceil$-subset of $E\setminus \left(\bigcup_{j<i}E_j\right)$.
This is easily implemented using the pseudorandom generator.

Consider item (4). Whenever the decremental certificate algorithm performs a query on the small boundary oracle, with high probability
the requested boundary $\partial_G(S)$ contains $O(c\cdot \polylog{n})$ edges.
As a result, only $O(c\cdot \polylog{n})$
edge bit-strings $x_f$ participate in each
computed value $y(j)=\bigoplus_{f\in \partial_G(S)\cap E_j}x_f$. Therefore,
it is enough that the edge bit-strings
are $O(c\polylog{n})$ independent instead of fully independent.\footnote{In general, the XOR trick can be used with polylogarithmic independence even for testing non-emptiness of large (i.e., up to size $n$) boundaries~\cite{DBLP:journals/corr/GibbKKT15, Thorup18}. 
Although a single bit of an edge bit-string
can be generated in constant time~\cite{Thorup18}, we need $\Theta(\log{n})$ bits per edge to guarantee high probability correctness. As a result, using known tools, generating all edge bit-strings would cost $\Theta(m\log{n})$ time which is too expensive for our application.} As a result, the individual bit-strings can be obtained from the pseudorandom generator in $O(m)$ time.

Finally, dealing with item~(2) is less straightforward.
This is because in the analysis of the data structure of~\cite[Theorem 6]{thorup1999decremental}, Thorup invokes a result due to Karger~\cite{Karger99} saying that
if a graph $G$ is $c'$-edge-connected graph,
where $c'=\Omega((c+\log{n})/p')$, and $p'\in (0,1)$
then $G(p')$ is $c$-edge-connected with high probability (depending on the constant hidden in the $\Omega$ notation).
This is then used to show that the number of
edges between different $c$-edge-connected components of a certificate (which is $G(p')$ augmented with the edges of $G$ connecting distinct $c$-edge-connected components of $G(p)$ , where $p'$ is any constant less than $1$)
that Thorup uses is $O(c'n)$ with high probability. This is where the $c+\log{n}$ term in the bounds in Theorem~\ref{t:thorup} comes from.
Unfortunately, roughly speaking, Karger's proof
applies a Chernoff bound to an \emph{exponential} number of cuts in $G$ and therefore requires sampling with full independence, i.e., $\Theta(m)$ random bits.

We can eliminate the need for full independence in~\cite{thorup1999decremental}, albeit at the cost of replacing the $c+\log{n}$ terms in the bounds of Theorem~\ref{t:thorup} with $c\log{n}$.
To this end, one can leverage Lemma~\ref{l:karger-alt} and replace the uniformly sampled subgraph $G(p')$ with the graph $H_\ell^0$ from our construction with $p$ set to $p'/\ell$ (computable in $O(mp\ell)=O(mp')$ time using $O(\polylog{n})$ random bits), at the cost of replacing the $c+\log{n}$ terms with $c\log{n}$ in the bounds of Theorem~\ref{t:thorup}. 

\section{Omitted Proofs}\label{sec:omitted-proofs}

\benczur*
\begin{proof}
	We proceed by strong induction on $n$. The statement is clearly true for $n=1, 2$. Consider now a graph $G$ on $n>2$ vertices and $(c-1)(n-1)+1$ edges. If no simple cut of $G$ of size $d<c$ exists, then $G$ is $c$-edge-connected and we are done. Otherwise, deleting the simple cut from $G$, we obtain subgraphs $G_1$ and $G_2$ of $G$ of sizes $n_1$ and $n_2$, respectively, such that $n_1+n_2=n$. After deleting the simple cut there are at least $(c-1)(n-2)+1=(c-1)(n_1-1)+(c-1)(n_2-1)+1$ edges left. Hence, by the pigeonhole principle and the induction hypothesis, either $G_1$ or $G_2$ contains a non-trivial $c$-edge-connected component.
\end{proof}
\ebce*
\begin{proof}
	Contracting the $c$-edge-connected components of $G$ we arrive at a graph $G'$ on $q_c$ vertices with no $c$-edge-connected components. All edges of $G$ that connect distinct $c$-edge-connected components of $G$ are still present in $G'$. By \cref{lem:many_edges_implies_con_comp}, $G'$ contains at most ${(c-1)(q_c-1)}$ edges, which completes the proof.
\end{proof}

\fddc*
\begin{proof}
    First of all, we store $G$ in a fully-dynamic connectivity data structure
    with $O(\log^2{n})$ amortized update time and $O(\log{n})$ query time~\cite{holm2001poly}.
    This data structure also maintains a spanning forest explicitly and allows
    $O(\log{n})$-time queries about the size of the component containing a given vertex. If an edge is inserted, we just pass the insertion to the fully-dynamic data structure -- this insertion does not change the connected components of $G$. If an edge $\{u,v\}$ is deleted, we additionally check if $u$ and $v$ are still connected after removing $\{u,v\}$. If not, assume wlog. that the component of $u$ is not smaller than that of $v$ afterwards.
    We set $A$ to be the vertices of the tree containing $v$ in the spanning forest. Let $q(u)=C_j$.
    For each $x\in A$ we remove $x$ from $C_j$, and set $q(x):=k+1$.
    Finally, we set $C_{k+1}:=A$, increment $k$, and output $(j,A)$.
    To bound the total time spent outside the fully-dynamic connectivity data
    structure, note that each time we spend time proportional to the
    size of the output set $A$. Whenever a vertex $x$ belongs to $A$,
    the size of the component of $x$ decreases by a factor of at least two
    due to the edge update. As a result, each $x$ can occur $O(\log{n})$ times
    in the output sets $A$, and hence the total size of these sets is $O(n\log{n})$.
\end{proof}

\fddb*
\begin{proof}
    Note that similarly as in Lemma~\ref{fully-dynamic-decrmental-components}, the updates to $\mathcal{C}$ are given in such a way that we can explicitly maintain, for each
    $v\in V$, the component from $\mathcal{C}$ it belongs to. This takes $O(n\log{n})$ total time. It also enables us to decide in $O(1)$ time
    whether an edge $\{u,v\}$ belongs to two (if $u,v$ are disconnected) or zero boundaries (otherwise).
    Each boundary $\partial_G(C)$ is stored in a linked list $L(C)$. Each edge of $G$ connecting endpoints in different components of $\mathcal{C}$ has associated
    two pointers to its places in the two respective lists.
    Hence, whenever an edge is inserted/deleted from $G$, the lists storing the boundaries can be easily updated in constant time.
    
    Now suppose $\mathcal{C}$ is updated: some $C_j\in\mathcal{C}$ gets split
    into $C_j\setminus A$ and $A$, where $|A|\leq |C_j\setminus A|$.
    Clearly, only the lists $L(C_j)$ and $L(A)$
    may need to be fixed at this point.
    We now iterate through all $\{u,v\}=e\in E_G(A,V)$ and proceed as follows.
    Suppose wlog. that $u\in A$. If $v\in C_j\setminus A$, then we add $e$ to $L(C_j)$ (it was not there before the split) and
    update all the auxiliary pointers.
    Otherwise, if $v\in V\setminus (C_j\cup A)$, then $e$ is removed from $L(C_j)$
    and inserted into $L(A)$.
    Otherwise, if $v\in A$, then we skip that edge as it remains an intra-component edge after the split.
    It is easy to verify that the lists represent the required boundaries
    after this step, which takes $O\left(\sum_{u\in A}\deg(u)\right)$ time.
    
    Finally, the $O(m\log{n})$ total update time bound follows since
    the incident edges of each vertex~$v$ are traversed $O(\log{n})$ times -- when
    this happens, the size of $v$'s component in $\mathcal{C}$ halves.
\end{proof}

\hldcert*
\begin{proof}
First, we prove that $H_\ell \cup D$ preserves the $c$-edge-connected components of $G$.
For convenience, denote $G_c = H_\ell \cup D$.
Assume, by contradiction, this is not true.
Note that for each $c$-edge-connected component $C'$ of $H_\ell$ it holds that $C' \subset C$ for some $c$-edge-connected $C$ component of $G$, as otherwise $G[C']$ contains a $<c$-cut and so does $H_\ell[C']$ since $H_\ell \subset G$; a contradiction.
Let $C$ be a $c$-edge-connected component of $G$ that is not preserved in $G_c$.
Then, there exists a $k$-cut $S$, for $k<c$, in $G_c[C]$ that is not a $k$-cut in $G[C]$.
Let $C_1,C_2$ be the two different connected components of $G_c[C] \setminus S$.
To conclude the argument, we next show that all edges in $(C_1 \times C_2) \cap E$ are present in $G_c$, which implies that if $S$ is a $k$-cut in $G_c[C]$ it is also a $k$-cut in $G[C]$, and hence we contradict the assumption that $C$ is a $c$-edge-connected component in $G[C]$ but not in $G_c[C]$.  
Take any edge in $uv \in (C_1 \times C_2) \cap E$.
Vertices $u$ and $v$ belong to different $c$-edge-connected components of $H_\ell$, as otherwise, there would be no $k$-cut separating $u,v$ in $G_c$ which contains $H_\ell$.
Hence, the edge $uv\in D\subseteq G_c$, since $D$ contains all edges of $G$ between components of $H_\ell$.
This concludes the proof that $G_c$ preserves the $c$-edge-connected components of $G$.

Now we turn to proving that $G_c$ preserves also the $c$-edge-connected classes of $G$. 
To this end, we first show that every $<c$-cut of $G_c$ is a $<c$-cut of $G$.
Let $S$ be a $<c$-cut of $G_c$.
For contradiction, suppose some $u,v\in V$ are connected in $G\setminus S$
but not in $G_c\setminus S$.
Let $P$ be a $u\to v$ path in $G\setminus S$.
The endpoints of some edge $xy\in P$ have to be disconnected
in $G_c\setminus S$, as otherwise
a path from $u$ to $v$ would exist in $G_c\setminus S$.
However, if $xy$ is contained in $G_\ell$, then $x$ and $y$
lie in the same $c$-edge-connected component of $H_\ell$, i.e., 
they are connected in $H_\ell\setminus S$ by $|S|<c$.
Otherwise, since $D=G\setminus G_\ell$, we have $xy\in D\setminus S$,
so $x$ and $y$ are connected in $G_c\setminus S$ as well.
This contradicts the fact that $x$ and $y$ are disconnected
in $G_c\setminus S$.

Now, let $u, v\in V$. If $u$ and $v$ are $c$-edge-connected in a subgraph of $G$, in particular $G_c$, then they are $c$-edge-connected in $G$. Conversely, if $u$ and $v$ are not $c$-edge-connected in $G_c$ then there is a cut of size $<c$ separating them in $G_c$. Such a cut is also a cut in $G$ by the previous claim, so $u$ and $v$ are not $c$-edge-connected in $G$ either. This proves that the $c$-edge-connected classes of $G$ and $G_c$ are identical.
\end{proof}

\newpage
\bibliographystyle{alpha}
\bibliography{bib}

\end{document}